\documentclass[10pt,reqno]{amsart}
\usepackage{hyperref}
\usepackage{float}
\usepackage{amscd,amssymb,amsmath,latexsym,bm}
\usepackage{amsfonts}
\usepackage{graphicx}
\usepackage{hyperref}
\usepackage{xcolor}
\usepackage{enumerate}

\usepackage{booktabs}
\usepackage{multirow,multicol}
 
\newtheorem{theorem}{Theorem}[section]

\newtheorem{corollary}[theorem]{Corollary}
\newtheorem{proposition}[theorem]{Proposition}
\newtheorem{remark}[theorem]{Remark}

\numberwithin{equation}{section}
\numberwithin{figure}{section}

\newcommand{\BM}{{\mathbb B}}
\newcommand{\CM}{{\mathbb C}}
\newcommand{\NM}{{\mathbb N}}
\newcommand{\QM}{{\mathbb Q}}
\newcommand{\RM}{{\mathbb R}}
\newcommand{\SM}{{\mathbb S}}

\newcommand{\ZM}{{\mathbb Z}}

\newcommand{\KM}{{\mathbb K}}

\newcommand{\Ss}{{\mathcal S}}

\newcommand{\Tt}{{\mathcal T}}

\newcommand{\Nn}{{\mathcal N}}

\newcommand{\Ll}{{\mathcal L}}

\newcommand{\Hh}{{\mathcal H}}

\newcommand{\PI}{{\rm \Pi}}

\newcommand{\braket}[2]{\langle #1|#2\rangle}        

\newcommand{\id}{e}

\newcommand{\fgfx}[1]{
  \raisebox{-0.4cm}{\scalebox{0.4}{\includegraphics{#1.png}}}
}

\begin{document}
\title{Topological Spectral Bands with Frieze Groups}

\author[F.R. Lux]{Fabian R. Lux}

\address{Department of Physics and
\\ Department of Mathematical Sciences 
\\Yeshiva University 
\\New York, NY 10016, USA}

\author[T. Stoiber]{Tom Stoiber}

\address{Department of Physics and
\\ Department of Mathematical Sciences 
\\Yeshiva University 
\\New York, NY 10016, USA}

\author[S. Wang]{Shaoyun Wang}

\address{Department of Mechanical and Aerospace Engineering, University of Missouri, Columbia, MO 65211, USA}

\author[G. Huang]{Guoliang Huang}

\address{Department of Mechanical and Aerospace Engineering, University of Missouri, Columbia, MO 65211, USA}

\author[E. Prodan]{Emil Prodan}

\address{Department of Physics and
\\ Department of Mathematical Sciences 
\\Yeshiva University 
\\New York, NY 10016, USA \\
\href{mailto:prodan@yu.edu}{prodan@yu.edu}}

\date{\today}

\begin{abstract} Frieze groups are discrete subgroups of the full group of isometries of a flat strip. We investigate here the dynamics of specific architected materials generated by acting with a frieze group on a collection of self-coupling seed resonators. We demonstrate that, under unrestricted reconfigurations of the internal structures of the seed resonators, the dynamical matrices of the materials generate the full self-adjoint sector of the stabilized group $C^\ast$-algebra of the frieze group. As a consequence, in applications where the positions, orientations and internal structures of the seed resonators are adiabatically modified,  the spectral bands of the dynamical matrices carry a complete set of topological invariants that are fully accounted by the K-theory of the mentioned algebra. By resolving the generators of the K-theory, we produce the model dynamical matrices that carry the elementary topological charges, which we implement with systems of plate resonators to showcase several applications in spectral engineering. The paper is written in an expository style.
\end{abstract}

\thanks{This work was supported by the U.S. National Science Foundation through the grants DMR-1823800 and CMMI-2131760, by U.S. Army Research Office through contract W911NF-23-1-0127 and the German Research Foundation (DFG) Project-ID 521291358.}

\maketitle

\setcounter{tocdepth}{2}

{\scriptsize \tableofcontents}

\section{Introduction}
\label{Sec:Introduction}

It has been long recognized that there is a reciprocal relation between an underlying discrete geometric pattern and the collective dynamical processes supported by that pattern, when the latter is populated with identical self-coupling resonators. It was Jean Bellissard who formalized this statement very precisely, when he introduced $C^\ast$-algebras and their K-theories to the condensed matter community \cite{Bellissard1986} in the context of atomic systems. His observation was that a material always carries a $C^\ast$-algebra, which, quite generally, can be defined as the enveloping $C^\ast$-algebra of the Hamiltonians generating the dynamics of material's degrees of freedom (both classical and quantum) under a precise set of specified experimental conditions. In a Galilean invariant theory, this $C^\ast$-algebra is often completely determined by the underlying discrete physical pattern. For example, Bellissard showed in \cite{Bellissard1986} that, for a material at finite temperature where thermal disorder breaks all crystalline symmetries, the Hamiltonians describing the low energy physics of the electrons can all be drawn from the stabilization\footnote{Stabilization means tensoring with the algebra $\KM$ of compact operators over a separable Hilbert space.} of a crossed product by $\RM^d$, where $d$ is the effective dimension of the material. Reciprocally, any self-adjoint element from this $C^\ast$-algebra can serve as the generator for the electrons' dynamics inside the material. Substantial contributions to this ingenious program were supplied by Kellendonk, who provided further insight into these intrinsic $C^\ast$-algebras \cite{KellendonkRMP95}.\footnote{Notably, the reduction of the crossed products to Morita equivalent groupoid $C^\ast$-algebras.}

An architected material is a synthetic material built with full control over their classical degrees of freedom, including their couplings. In this context, we can choose how to define specific classes of materials as well as the set of experimental conditions the materials will be subjected to. For example, often in the typical applications,  identical resonators are placed in a give architecture and their internal structures ({\it i.e.} mode shapes and force fields) are unrestrictedly modified in order to explore or exploit all dynamical effects offered by the underlying architecture.  In such conditions, the minimal $C^\ast$-algebra covering the dynamical matrices supported by an architecture can be explicitly computed \cite{MeslandJGP2024} in the form of a groupoid $C^\ast$-algebra built entirely and canonically from the architecture of the material. Now, if the $C^\ast$-algebras of two architected materials are found to be isomorphic, then any dynamical feature observed for one material can be reproduced with the other. In other words, from the dynamical point of view, the two architectures are identical. When the architected materials are approached from this angle, the following simple but deep principle emerges: Discovering new dynamical patterns in metamaterials amounts to exploring different stably isomorphic\footnote{Stably isomorphic algebras are said to be Morita equivalent.} classes of $C^\ast$-algebras (see \cite{MeslandJGP2024} for more on this point of view). The present work advances this program by investigating classes of patterns associated with the discrete groups of symmetries of a strip, known as frieze groups \cite{ConwayBook} (see section~\ref{Sec:SymmetricPatterns}).

The operator algebraic approach with its K-theoretic tools introduced by Bellissard has much more to offer \cite{Bellissard1986,BellissardJMP1994,
Bellissard1995,Bellissard2000,Bellissard2003}. Indeed, stably isomorphic or Morita equivalent $C^\ast$-algebras have identical K-theories. Therefore, the K-theoretic groups can be regarded as topological invariants for the deformations of architected materials within one class. Thus, after restricting to a well defined class of architected materials and experimental conditions, therefore to a fix algebra of dynamical matrices, each connected component of the spectrum of a dynamical matrix, aka spectral band, can be labeled by an element of the $K_0$-group of the algebra. This element represents the complete topological invariant associated to the corresponding spectral projection, called band projection from now on. When the abelian $K_0$-group is presented in a basis, then the complete topological invariant breaks down into a finite set of topological charges (see Sec.~\ref{Sec:TopologicalDynamics}). If these topological charges are found to be identical for all band projections of two different dynamical matrices, then we can be sure that the two dynamical matrices can be stably deformed into each other, inside the fixed algebra, without changing the topology of the spectrum. However, if the topological charges are found to be different, then the topology of the spectrum must change during the deformations and, as such, some of the spectral gaps must close, regardless how hard we try to avoid this phenomenon. This supplies a very simple, robust and practical principle for resonant spectrum engineering of metamaterials, because it supplies the means to close and open spectral gaps on demand, without fine-tuning. Furthermore, it gives us insight into the dynamical features that can be achieved with a material from a given a class and in a given set of experimental conditions, in the sense that all such materials deliver wave-channels that carry a finite number of topological charges that are conserved when bands collide or split as a result of continuous deformations of the material. In fact, if one is interested only in dynamical effects that are robust against continuous deformations of the materials, referred to here as topological dynamics, then all such dynamical effects can be reproduced by stacking elementary models that carry the fundamental topological charges. 

To put the above principles at work in real world applications, one needs to: 
\begin{enumerate}[\ \ $N1$.]
\item Specify the protocol generating the class of materials; 
\item Specify the allowed deformations of the materials; 
\item Compute the corresponding unique $C^\ast$-algebra of dynamical matrices; 
\item Compute the $K_0$-group of this algebra, together with an explicit basis; 
\item Supply the model dynamical matrices corresponding to this basis. 
\end{enumerate}
If all these components are in place, then a theorist can advise an experimental research group on how to reconfigure a meta-material to produce dynamical effects that are stable against deformations of the materials, such as spectral flows that close and open specific spectral gaps. In a different scenario, one may be dealing with a complicated dynamical matrix coming from a laboratory, in which case one will be interested to resolve the complete topological invariant of the band projections. In this scenario, one needs:

\begin{enumerate}[\ \ $N6.$]
\item A practical algorithm to compute the $K$-theoretic labels. 
\end{enumerate} 

The structures of the frieze groups are relatively simple and the program outlined above can be completed in its entirety. This created for us the opportunity to present here a model of analysis that could be of some guidance for mathematical physicists seeking collaborations with metamaterial engineers and, reciprocally, for material engineers seeking inspiration from this type of works. Along the way, we intend to advertise several tools specific to operator algebras, which not only offer effective vehicles for computations but also natural frameworks that self-explain the purpose of a calculation and, in the same time, guide one with what needs to be computed. Specifically: 1-2) We introduce an algorithm that generates a well defined class of architected materials by acting with the space transformations of a frieze group on a set of seeding resonators. 3) Under the assumption that the interactions between pairs of resonators are fully determined by their relative geometric configuration, as it is always the case if the physics involved is Galilean invariant, we show that the dynamical matrices governing the collective dynamics of the resonators belong to the stabilized group $C^\ast$-algebra of the corresponding frieze group. Reciprocally, in experimental conditions where the internal structures, positions and orientations of the seed resonators are unrestrictedly modified, than the self-adjoint sector of this algebra is fully sampled by the dynamical matrices of the architected materials. 4-5) We supply a complete account of the $K_0$-theory of this algebra for all seven classes of frieze groups, together with explicit sets of generators of the $K_0$-groups, as well as model dynamical matrices. In the process, we advertise the Baum-Connes machinery \cite{BaumIHES1982,BaumProceedings1988,Baum1987,
BaumContMath1994,BaumEM2000}. 6) For a generic dynamical matrix, we also supply a numerical algorithm for computing the complete set of K-theoretic labels of the band projectors. In the process, we advertise Kasparov's bivariant K-theory and its internal product \cite{KasparovJSM1981,KasparovJSM1987} as the natural framework to conceptualize the process of resolving the K-theoretic labels. The resulting algorithm never makes appeal to a Bloch decomposition since it is entirely developed inside the real-space representation. Among other things, we believe this approach can handle disordered perturbations. 7) Lastly, we implement the model dynamical matrices using phononic crystals and we use COMSOL\footnote{COMSOL Multiphysics is a widely adopted commercial platform for finite element analysis, known to produce reliable simulations of continuous mechanics models.} simulations to demonstrate various topological spectral flows generated with the principles discussed above. Let us emphasize that we make no attempts to enforce any of the fundamental symmetries, that is, the time-reversal, particle-hole or chiral symmetries, which require special experimental conditions (see \cite{BarlasPRB2018}).

We now discuss the relation between our work and the existing literature. We start with a brief survey of the physics literature. The first indication that a space symmetry can enrich the topological dynamics of a material appeared in \cite{FuPRL2011}. The first indication that a space symmetry {\it alone} can stabilize topological phases in a material came from Refs.~\cite{TurnerPRB2010,HughesPRB2011}. Prior to these works, a topological phase was synonymous with a non-trivial (strong) bulk-boundary correspondence, but an alternative definition was put forward in \cite{TurnerPRB2010,HughesPRB2011}, saying that a band insulator is in a topological phase if it cannot be adiabatically deformed to its atomic limit. This concept evolved in what is today known as topological quantum chemistry \cite{PoNatComm2017,BradlynNature2017,CanoPRL2018,VergnioryPRE2017}, which perhaps can be defined as the science of identifying topological bands in stoichiometric condensed substances. It comes with fine tools that have been combined with first-principle computer simulations of quantum solids to assess the topology of the energy bands in large classes of stoichiometric materials \cite{WiederScience2018,VergnioryNature2019,ElcoroNatComm2021,
WiederNatRevMat2022,VergnioryScience2022}. The topological criterium in these works is the pure homotopy equivalence of the bands, while ours is stable homotopy. The latter allows new degrees of freedom to participate in the deformation processes of the models and, as a result, our predictions hold even when the physical systems interact with external structures, such as supporting frames. We want to make it clear, though, that we are not interested here in the classification of topological phases but rather in the dynamics of small well defined classes of materials. More precisely, we want to make ourself useful to the metamaterial community by developing bottom-up design approaches that deliver specific topological dynamics in concretely specified experimental conditions, as an alternative to the top-bottom approach in \cite{WiederScience2018,VergnioryNature2019,ElcoroNatComm2021,
WiederNatRevMat2022,VergnioryScience2022}, where large libraries of existing materials were scrutinized for topological bands. Of course, the latter still represent an extremely valuable source of information for meta-material scientists.

We now turn our attention to the mathematical physics literature, specifically to the works by Shiozaki et al \cite{ShiozakiPRB2014,ShiozakiPRB2016,ShiozakiPRB2017}, which initiated the systematic applications of topological K-theoretic methods to systems with space group symmetry. These works include the enrichment by the fundamental symmetries and treat all Wyckoff positions at once. It is shown there that, after a Bloch decomposition, the spectral bands are classified by twisted equivariant K-theories (as introduced in \cite{FreedAHP2013}; see also \cite{GomiSIGMA2017}) over the Brillouin torus w.r.t. to actions of the point group. These K-theories are carefully formulated in generic setting and \cite{ShiozakiPRB2017} explicitly computes them for all 17 wallpaper groups in class A and class AIII. The results revealed an extremely rich landscape of topologically distinct phases, which can be a great source of inspiration not only for electronic materials but also for metamaterials where dynamics is carried by classical degrees of freedom. Refinements, extensions and subtle new point of views have been further provided in \cite{GomiLMP2019,GomiJGP2019,GomiIJM2021,ShiozakiPTEP2023}, just to mentioned a few works from a rapidly growing body of rigorous publications (see also \cite{ShiozakiArxiv2023} for an impressive tour de force). Certainly, the field is at a point where massive and effective deployment of these abstract predictions might be witness in laboratories in the coming years (see the frieze acoustic crystals in \cite{MooreJASA2024}). The purpose of our paper is precisely to demonstrate how this could happen in the context of architected materials. As already inferred in our list $N1$-$N6$, accent is put on what kind of information and tools are needed to make that happen and what could be the practical issues that need attention at the fundamental level. For example, one such issue is the stability of the predictions in the presence of fabrication imperfections. We foresee the operator algebras and operator K-theory, as opposed to topological K-theory, as the vehicles to put this issue under control. Furthermore, as projected by our list $N1$-$N6$, establishing the isomorphism classes of the K-groups is only a small piece of the information needed to produce results in a laboratory. Even for the simple case of frieze groups, the second part of $N4$, $N5$ and $N6$ are missing or are scattered in the published literature.

The pure mathematics literature abounds with results on the K-theories of group $C^\ast$-algebras. One of the most effective tools of computations is the Baum-Connes assembly map \cite{BaumIHES1982,BaumProceedings1988,Baum1987,
BaumContMath1994,BaumEM2000}, which more often than not reduces the task to computing equivariant K-homologies of topological spaces (see \cite{AparicioANC2019} for a status report and \cite{GomiIJM2021} for a direct application related to ours). Davis and L\"uck \cite{DavisLueck} have unified several existing assembly maps, which then enabled L\"uck and Stamm \cite{LueckStamm} to formally derive the K-theories of all crystallographic groups in arbitrary dimensions and, explicitly, the K-theories of all wallpaper groups (see also \cite{YangThesis}). While we hope that these results and methods will soon become a major source of inspiration in materials science, we will use them here only as a reference, as we will showcase simpler but less powerful computations based on the equivariant Chern character \cite{Baum1987,Lueck2002} (see Th. 6.1 of \cite{Mislin03}), which can only resolve the non-torsion component of the $K_0$-groups.\footnote{This is entirely sufficient in our case.} As already pointed out in \cite{Echterhoff2010}, the Baum-Connes machinery is not always an effective tools in this respect. For crystallographic groups of low dimensions, the methods based on non-commutative CW-complexes developed in \cite{McAlisterPhD2005} are capable to deliver both the $K_0$-groups and sets of generators. Other more direct methods for specific crystallographic groups can be found in \cite{Cuntz82,HadfieldJOP2004}. We will use some of these sources when listing the generator of the $K_0$-groups of the frieze group $C^\ast$-algebras. While this is sufficient for the simple frieze groups invoked by our study, more systematic methods for resolving the bases of the $K_0$-groups are yet to be developed for the space groups in higher dimensions.

\section{A Class of Architected Materials}
\label{Sec:SymmetricPatterns}

A strip is a quasi one-dimensional slice of the Euclidean plane, such as $\Ss := \RM \times [-1,1]$, equipped with the inherited Euclidean metric. 
Its complete group of isometries is generated by continuous translations along the long axis and reflections relative to the horizontal and vertical axes. This continuous group of isometries accepts seven distinct discrete subgroups, the frieze groups \cite{ConwayBook}. This section supplies a brisk review of the frieze groups and introduces our specific algorithm that creates architected materials using actions of the frieze groups. The section also analyzes the dynamical matrices of these materials.

\subsection{Discrete symmetries of a strip}\label{Sec:frieze} The seven discrete subgroups of the full group of isometries of the strip are all listed in table~\ref{Ta:1}, together with their isomorphism classes, sets of generators, and actions on the strip. 
In this table, $\tau$ is the translation by one unit along the strip, while $\sigma_h$ and $\sigma_v$ are the reflections against the horizontal and vertical axes of the strip, respectively.

\vspace{0.1cm}

{\small
\begin{table}
  \caption{\small List of the frieze groups, their isomorphism classes, generators and actions on the plane.}
  \centering
  \begin{tabular}{ccclccc}
  \toprule
  Iso. class & \ \ \ \  $\Gamma$ \ \ \ \ & \ \ \ \ Generators \ \ \ \ & Orbifold    \\ \midrule  \addlinespace[0.5em]
  $\mathrm{Z}_\infty$  & $p1$ &   $\tau$ & \fgfx{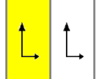}     \\  \addlinespace[0.5em]
  & $p11g$ &  $\tau_{1/2} \sigma_v$ & \fgfx{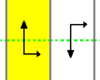}  \\ \addlinespace[0.5em] \midrule
 $\mathrm{Dih}_\infty$  & $p1m1$ & $\tau,\sigma_v$ & \fgfx{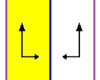}    \\ \addlinespace[0.5em]
  & $p2$ & $\tau,\sigma_h \sigma_v$ & \fgfx{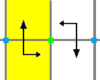}   \\ \addlinespace[0.5em]
  & $p2mg$ & $\tau_{1/2}\sigma_h, \sigma_v$ & \fgfx{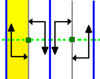}  \\ \midrule \addlinespace[0.5em]
  $\mathrm{Z}_\infty \times \mathrm{Dih}_1$  & $p11m$ & $\tau, \sigma_h$ & \fgfx{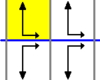}   \\ \midrule \addlinespace[0.5em]
  $\mathrm{Dih}_\infty \times \mathrm{Dih}_1$   & $p2mm$ & $\tau, \sigma_h, \sigma_v$ & \fgfx{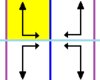}   \\ \midrule \addlinespace[0.5em]
  \end{tabular}
  \label{Ta:1}
\end{table}
}

The structure of the frieze groups can be summarized in a concise form by presenting each of them in terms of their generators and relations:
\begin{enumerate}
\item $ p1  = \langle u \rangle,  \ u = \tau;$
\item $p11g  = \langle u \rangle,  \ u = \tau_\frac{1}{2} \sigma_v;$ 
\item $p1m1  = \langle u,v_1 ~|~ v_1^2, \ (u v_1)^2\rangle, \ u = \tau, \ v_1 =\sigma_v;$
\item $p2 = \langle u,v_1 ~|~ v_1^2, \ (u v_1)^2\rangle, \ u = \tau, \ v_1=\sigma_h\sigma_v$;
\item $p2mg = \langle u,v_1 ~|~ v_1^2, \ (u v_1)^2\rangle, \ u = \tau_{1/2} \sigma_h, \  v_1 = \sigma_v$;
 \item $p11m  = \langle u,v_2 ~|~ v_2^2, \ u v_2 u^{-1} v_2 \rangle, \ u = \tau, \ v_2 =\sigma_h;$
 \item $p2mm  = \langle u,v_1,v_2 ~|~ v_j^2, (v_1 v_2)^2, (u v_1)^2, u v_2 u^{-1} v_2\rangle,  \ u = \tau, \ v_1 = \sigma_v, \  v_2 = \sigma_h.$
\end{enumerate}
Thus, the discrete subgroups fall into four isomorphism classes, as already indicated in Table~\ref{Ta:1}. Furthermore, we can use a uniform notation to specify an element of any of the frieze groups as $u^n v_1^{\alpha_1} v_2^{\alpha_2} : = u^n v^\alpha$, where it is understood that, for example for $p1m1$ group, $\alpha_2$ takes only value 0. Here $\alpha=(\alpha_1,\alpha_2)$ and $v^\alpha = v_1^{\alpha_1} v_2^{\alpha_2}$. In this presentation, the composition rule is simply
\begin{equation}\label{Eq:Compo1}
\big (u^n v^\alpha\big) \cdot \big ( u^m v^\beta \big) = u^{n+(-1)^{\alpha_1} m} \; v^{(\alpha + \beta){\rm mod}2},
\end{equation}
and inversion
\begin{equation}\label{Eq:Inversion}
(u^n v^\alpha)^{-1} = u^{(-1)^{\alpha_1}} v^\alpha.
\end{equation}
The above rules supply a practical way to encode the tables of the frieze groups, which will be employed by our numerical simulations as explained in section~\ref{Sec:Numerical}.

\subsection{The class of architectures and the assumed experimental conditions}\label{Sec:Gen} The protocol that defines the architected materials we want to investigate consists of the following steps:
\begin{enumerate}[P1.]
\item Start with a finite number $N_s$ of seeding discrete resonators, placed in a strip at desired non-overlapping locations and with desired orientations;
\item Apply a plane transformation from a chosen frieze group to each of the seed resonators.
\item Replenish the seed resonators and repeat step two for all the plane transformations contained in the frieze group.
\item Adjust the seeds if steps 1-3 result in overlapings and repeat the process.
\end{enumerate}

\begin{figure}
\center
  \includegraphics[width=0.99\linewidth]{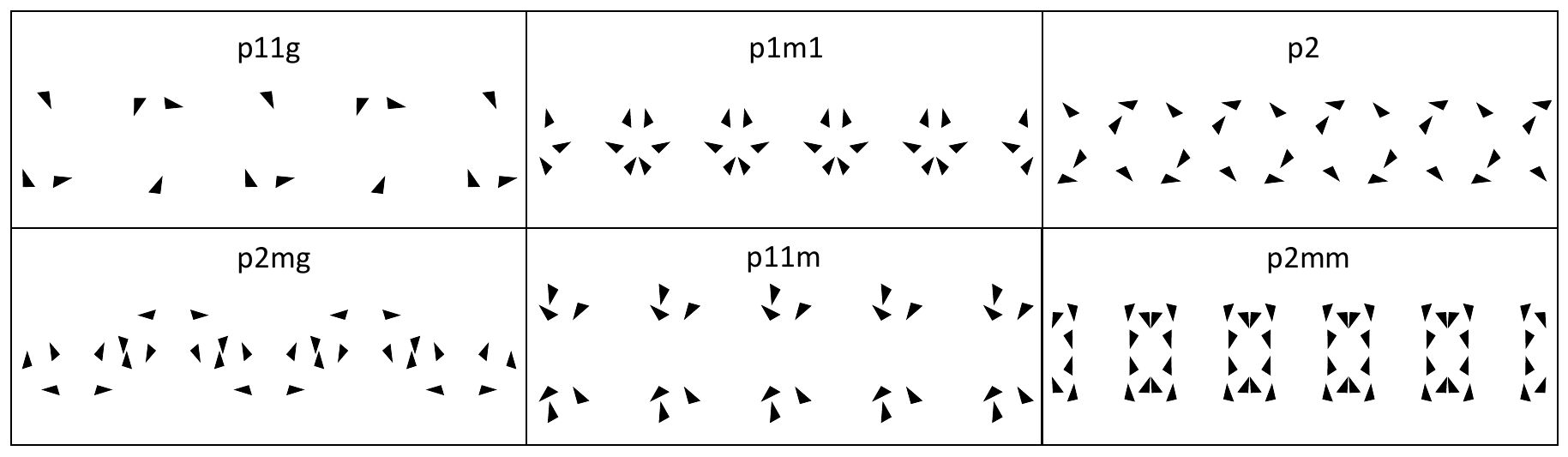}
  \caption{\small Examples of architectures generated from three seed resonators.}
 \label{Fig:1}
\end{figure}

Examples of symmetric architectures generated from three seed triangles are shown in Fig.~\ref{Fig:1}. 

\begin{remark}{\rm There are several observations to be made about {\it our} class of materials. Firstly, the resulting architectures are symmetric w.r.t. the frieze group engaged by the protocol. However, as stated in our last rule, we will be avoiding those particular seed configurations that lead to overlappings and, as such, the seeds are always placed at general Wyckoff positions (i.e. those points of the strip that are fixed only by the action of the identity operation of the frieze group). Of course, since the number of seeds can be arbitrarily large, a finite number of them can assemble into a super-seed with center at a special Wyckoff position. Still, this is a special situation because the site group of that Wyckoff position acts freely on the local modes. The bottom line is that, by choice, there are symmetric architectures that are not included in {\it our} class of materials. If we insist in including those, then the algebra of dynamical matrices, for those particular cases, drops from a group to a groupoid $C^\ast$-algebra \cite{MeslandJGP2024}.
}$\Diamond$
\end{remark}

We also need to specify the experimental conditions. For this, note that any of the proposed architected materials is completely determined by the shape, position, orientation and internal structures of the seeding resonators. We will assume that the configuration space defined by the factors we just listed is fully and unrestrictidely explored during the experimental applications, within the bounds imposed by the protocol P1-P4.

\subsection{Dynamical matrices} The resonators  can be steel plates embedded in an elastic medium, acoustic cavities, or magnetized solid bodies (see Example~2.6 in \cite{MeslandJGP2024}, \cite{MooreJASA2024} and section~\ref{Sec:PlateResonators} for concrete examples). Each resonator in the architecture can be conveniently and uniquely labeled by $(i,f) \in \{1,\ldots,N_s\} \times F$, where $i$ corresponds to the seed resonator on which the transformation $f$ was applied. We emphasize that this is a consequence of our decision to exclude any overlap of the resonators. Each resonator carries a finite number $K$ of relevant internal resonant modes. These modes will always be observed, measured and quantified in a frame rigidly attached to the resonator. As a result, these observations and measurements are insensitive to the Euclidean transformations applied on the resonators and this enables us to choose the bases for the internal spaces of resonators in a coherent fashion. Specifically, for the seed resonator $i$, we choose once and for all a basis for its space of internal resonant modes, which we denote by $|\alpha,i,e\rangle$, with $\alpha =\overline{1,K}$ and $e$ the neutral element of the frieze group $F$. When the seed resonator is acted with a plane transformation $f$ from the frieze group and placed to its rightful configuration in the architectures, the internal space can be rendered in the basis $|\alpha,i,f\rangle$, $\alpha =\overline{1,K}$, which coincides with the modes $|\alpha,i,e\rangle$ when observed from the intrinsic frame. 

\begin{remark}{\rm This coherent labeling of the modes always exist for our special class of materials. For atomic systems, however, the orbitals are always observed and quantified in the laboratory frame, hence quite differently from what we are proposing here for classical resonators.
}$\Diamond$
\end{remark} 

The collective dynamics of a pattern of resonators takes place inside the Hilbert space $\CM^N \otimes \ell^2(F)$, $N = N_s \, K$, spanned by the vectors $\xi \otimes |f\rangle$ with $\xi \in {\rm Spann}\{|\alpha,i\rangle, \ \alpha=\overline{1,K}, \ i = \overline{1,N_s} \} \simeq \CM^N$ and $f \in F$. Given our convention, a plane transformation from the frieze group acts as 
\begin{equation}\label{Eq:PlaneTf}
T_{f'} (\xi \otimes |f\rangle) = \xi \otimes |f' f \rangle .
\end{equation}
It is important to note that the vector $\xi$ is not acted on by the transformation. Now, like any bounded operator over $\CM^N \otimes \ell^2(F)$, a dynamical matrix assumes the generic form
\begin{equation}\label{Eq:GenD}
H = \sum_{f,f' \in F} w_{f,f'} \otimes |f \rangle \langle f' |, \quad w_{f,f'} \in M_{N}(\CM),
\end{equation}
where $M_{N}(\CM)$ denotes the algebra of $N \times N$ matrices with complex entries. As an experimental fact, the coupling matrix $w_{f,f'}$ becomes weaker as the distance between a pair $(f,f')$ of resonators is increased and in fact the coupling matrix cannot be experimentally resolved beyond a certain separation distance. Thus, for typical experiments, it is justified to assume that there is a finite number of terms in Eq.~\ref{Eq:GenD} involving one $f \in F$.

\vspace{0.1cm}

If the physical processes involved in the coupling of resonators are Galilean invariant, then the coupling matrices must satisfy the following symmetry constraints
\begin{equation}\label{Eq:Invariance}
w_{g \cdot f, g \cdot f'} = w_{f,f'}, \quad \forall \ f,f',g \in F.
\end{equation}
It is quite important to understand that this is not an assumption but rather a physical reality.  Indeed, let  there be some underlying Galilean invariant physical laws $\Ll \mapsto \{w_{f,f'}(\Ll)\}$ that associate to each admissible architecture $\Ll$ a family of coupling matrices. Then, if $g$ is any Euclidean transformation, we must have
\begin{equation}
w_{g \cdot f, g \cdot f'}(g \cdot \Ll) = w_{f,f'}(\Ll),
\end{equation}
because any instrumentation rigidly attached to the frame of the resonators will not be able to sense the Euclidean drift $g$ of the entire platform. Now, if $g$ belongs to the frieze group that produced the pattern in the first place, then $g \cdot \Ll = \Ll$ and the Eq.~\eqref{Eq:Invariance} follows.

\begin{remark}{\rm As opposed to atomic systems, no additional conjugation with a representation of $g$ on the space $\CM^N$ is needed in Eq.~\eqref{Eq:Invariance}, because the pair of resonators and their resonant modes are rigidly rotated by $g$ and the modes are observed and quantified from frames rigidly attached to the resonators. Another way to justify the absence of such conjugation is via the specific action~\eqref{Eq:PlaneTf} of the plane transformations on our basis.
}$\Diamond$
\end{remark} 

As a direct consequence of Eq.~\eqref{Eq:Invariance}, any dynamical matrix over one of our symmetric patterns can be always reduced to the following particular form
\begin{equation}\label{Eq:ReducedH}
H = \sum_{f,f' \in F} w_{1,f^{-1} \cdot f'} \otimes |f \rangle \langle f' |.
\end{equation}
This is already a strong indication that the dynamical matrices of our patterns of resonators form a sub-algebra of the bounded operators over $\CM^N \otimes \ell^2(F)$. Computing this sub-algebra is the subject of the next section.

\begin{remark}{\rm The algebra $\BM\big (\ell^2(F)\big )$ of bounded operators over $\ell^2(F)$ is very large and it has a trivial K-theory. Thus, in order to predict topological phenomena in the dynamics of the resonators, it is paramount to demonstrate that the algebra of dynamical matrices is actually a separable $C^\ast$-subalgebra of $\BM\big (\ell^2(F)\big )$.
}$\Diamond$
\end{remark}

\section{The $C^\ast$-algebras of dynamical matrices}
\label{Sec:GCAlg}

This section establishes the connection between the abstract $C^\ast$-algebra of the frieze groups and the algebra of dynamical matrices of the class of architected materials introduced in the previous section.

\subsection{Group $C^\ast$-algebras: Elementary facts and notations}\label{Sec:CG} There will be several groups involved in our analysis and, for this reason, we want to introduce some uniform notation that can be used exchangeably among the groups. The standard material of this subsection can be found, for example, in \cite{DavidsonBook}[Ch.~VIII].

\vspace{0.1cm}

Henceforth, given a generic discrete group $G$, its group algebra $\CM G$ consists of formal series
\begin{equation}\label{Eq:CG}
q = \sum_{g \in G} \alpha_g \, g, \quad \alpha_g \in \CM,
\end{equation}
where all but a finite number of terms are zero. 
The addition and multiplication of such formal series are defined in the obvious way, using the group and algebraic structures of $G$ and $\CM$, respectively. 
In addition, there exists a natural $\ast$-operation
\begin{equation}
q^\ast = \sum_{g \in G} \alpha_g^\ast \, g^{-1}, \quad (q^\ast)^\ast = q, \quad (\alpha q)^\ast = \alpha^\ast q^\ast, \ \alpha \in \CM.
\end{equation}
Hence, $\CM G$ is a $\ast$-algebra in a natural way.

\vspace{0.1cm}

We denote by $e$ the neutral element of $G$. Then the map
\begin{equation}
\Tt : \CM G \to \CM, \quad \Tt(q) = \alpha_e
\end{equation}
defines a faithful positive trace on $\CM G$ and a pre-Hilbert structure on $\CM G$ via
\begin{equation}
\langle q, q' \rangle : = \Tt(q^\ast q'), \quad q,q' \in \CM G.
\end{equation}
The completion of the linear space $\CM G$ under this pre-Hilbert structure supplies the Hilbert space $\ell^2(G)$, spanned by $|g\rangle$, $g \in G$, which form an orthonormal basis
\begin{equation}
\langle g , g' \rangle = \delta_{g,g'}, \quad g,g' \in G.
\end{equation} 
The left action of $\CM G$ on itself can be extended to the action of a bounded operator on $\ell^2(G)$, and this supplies the left regular representation $\pi_L$ of $\CM G$ inside the algebra $\BM\big(\ell^2(G)\big)$ of bounded operators. Specifically,
\begin{equation}\label{Eq:PiLeftReg}
\pi_L(q) |g\rangle = \sum_{g'\in G} \alpha_{g'} |g' g \rangle = \sum_{g'\in G} \alpha_{g'g^{-1}} |g' \rangle, \quad q \in \CM G, \quad g \in G,
\end{equation} 
where $q = \sum_{g' \in G} \alpha_{g'}  g'$ with $\alpha_{g'} \in \CM$.
The completion of $\CM G$ with respect to the norm
\begin{equation}\label{Eq:CNorm}
\|q\| : = \|\pi_L(q)\|_{\BM(\ell^2(G))}
\end{equation}
supplies the reduced group $C^\ast$-algebra $C^\ast_r(G)$ of $G$.

\begin{remark}{\rm As we shall see, the stabilization of this group algebra represents the space where the dynamical matrices live and are deformed. The norm~\eqref{Eq:CNorm} puts a topology on this space, which gives a rigorous meaning to the wording ``continuous deformations'' used up to this point in an intuitive but very imprecise mode. 
}$\Diamond$
\end{remark}

\begin{remark}{\rm For any state on a $\ast$-algebra, there are left and right GNS representations. The right GNS representation $\pi_R$ corresponding to $\Tt$ has the same Hilbert representation space, but the elements act from the right (see Eq.~\eqref{Eq:PiRightReg}).
}$\Diamond$
\end{remark}

\begin{remark}\label{Re:TT}{\rm The following detail is worth mentioning. If the group $G$ is finite, then
\begin{equation}
\Tt(q) = \frac{1}{|G|} {\rm Tr}\big ( \pi_L(q) \big ) = \frac{1}{|G|} {\rm Tr}\big ( \pi_R(q) \big ),
\end{equation}
which can be seen straight from \eqref{Eq:PiLeftReg}.
}$\Diamond$
\end{remark}

Since we only deal here with discrete amenable groups, the reduced and full group $C^\ast$-algebras coincide. This is important because the relation between the group and its $C^\ast$-algebra is functorial only for the latter. So, in the present context, we have the following simple but essential statement:

\begin{proposition}[\cite{NassopoulosIJCMS2008}]\label{Pro:LiftMor} Let $G$ and $H$ be discrete amenable groups and $\rho : G \to H$ be a group morphism. Then $\rho$ lifts to a $C^\ast$-algebra morphism $C^\ast_r(G) \to C^\ast_r(H)$, extending the map between group algebras
\begin{equation}
\CM G \ni \sum_{g \in G} \alpha_g \, g \mapsto \sum_{g \in G} \alpha_g \, \rho(g) \in \CM H.
\end{equation}
\end{proposition}

\subsection{Symmetry and structure of dynamical matrices} As we have seen, the group $G$ acts on itself via an action either from the left or from the right. The symmetry operations associated to the group elements are implemented via the left-regular representation:
\begin{equation}
U : G \to \BM\big ( \ell^2(G) \big ), \quad U_g |g' \rangle : = \pi_L(g)|g'\rangle = |gg'\rangle.
\end{equation}
The right regular representation of the group algebra acts as
\begin{equation}\label{Eq:PiRightReg}
\pi_R(q) |g'\rangle : = \sum_{g\in G} \alpha_{g} |g' g^{-1} \rangle = \sum_{g\in G} \alpha_{g^{-1} g'} |g \rangle, \quad q \in \CM G, \quad g' \in G.
\end{equation} 
Since the left and right group actions commute, we automatically have that
\begin{equation}
U_g^\ast \, \pi_R(q) \, U_g = \pi_R(q), \quad q \in C^\ast_r(G), \quad g \in G.
\end{equation}
Hence, the group $C^\ast$-algebra supplies models that are naturally symmetric w.r.t. the group $G$, via a mechanism that rests entirely on the associative property of the group multiplication. 

\begin{remark}{\rm The following simple observation establishes a useful link between projections and representations. Suppose $p$ is a projection from $C^\ast_r(G)$. Then $\pi_R(p)\ell^2(G)$ is an invariant subspace for the group symmetries
\begin{equation}
U_g \big [ \pi_R(p)\ell^2(G) \big ] = \pi_R(p)\ell^2(G), \quad \forall \ g \in G.
\end{equation}
As such, $g \mapsto  \pi_L(g) \pi_R(p)$ supplies a representation of $G$ on the subspace $\pi_R(p)\ell^2(G)$. This simple observation will be exploited here quite often.
}$\Diamond$
\end{remark}

We now specialize the discussion to the case when $G$ is the frieze group generating a class of architected materials. The  right regular representations of $C^\ast_r (F)$,
\begin{equation}
\pi_R\Big (\sum_f \alpha_f \, f\Big ) = \sum_{f,f' \in F} \alpha_{f^{-1}f'} |f\rangle \langle f'|,
\end{equation}
can be naturally extended to a representation of the algebra $M_N(\CM) \otimes C^\ast_r(F)$ over the Hilbert space $\CM^N \otimes \ell^2(F)$. Specifically, an element from $M_N(\CM) \otimes C^\ast_r(F)$ takes the same form $\sum_f \alpha_f \, f$, but with coefficients $\alpha_f$ from $M_N(\CM)$, and the extension of the right regular representation is 
\begin{equation}\label{Eq:RRRep}
\pi_R\Big (\sum_f \alpha_f \, f\Big ) = \sum_{f,f' \in F} \alpha_{f^{-1}f'} \otimes |f\rangle \langle f'|.
\end{equation}
Similarly, the group symmetries $U_f$ can be lifted as $I \otimes U_f$ on the Hilbert space $\CM^N \otimes \ell^2(F)$ by letting $F$ act trivially on $\CM^N$.

Now, comparing the generic form Eq.~\eqref{Eq:ReducedH} of the dynamical matrices for Galilean invariant theories with the right regular representation~\eqref{Eq:RRRep}, we see that the two coincide if we take $w_{f,f'}=\alpha_{f^{-1}f'}$. Furthermore, when this correspondence is applied on elements from $\CM F$, it supplies dynamical matrices with finite coupling ranges. The important conclusion is that {\it all} dynamical matrices of the patterns of resonators proposed by us can be generated from the right regular representation of $M_N(\CM) \otimes C^\ast_r(F)$, for some finite but otherwise arbitrary $N$, and, reciprocally, the dynamical matrices that can be engineered in a laboratory densely sample the self-adjoint sector of $\KM \otimes C^\ast_r(F)$. Furthermore, we can be sure that our allowed deformations of the physical systems, which involve changing the shapes, positions and orientations, as well as the internal structures of the seeding resonators, take place inside the stabilization of the $C^\ast$-algebra $C^\ast_r(F)$.

\begin{remark}{\rm Throughout, we will reserve the symbols $\pi_R$, $\Tt$ and $U$ for the objects introduced in this section, making sure that the group $G$ can be clearly read off from the context. In fact, since the left regular representation is only used in tandem with $U$, we simplify the notation for the right regular representation to $\pi$.
}$\Diamond$
\end{remark}

For the practical purposes of this paper, it is useful to merge the general discussion and the specifics of the frieze groups presented in subsection~\ref{Sec:frieze}. Using the uniform notation introduced there, the elements of $M_{N}(\CM) \otimes C^\ast_r(F)$ can be presented as norm-convergent infinite series
\begin{equation}\label{Eq:Elem1}
\sum_{n \in \ZM} \sum_{\alpha \in \{0,1\}^2} w(n,\alpha) u^n v^{\alpha},
\end{equation}
with coefficients $w(n,\alpha)$ from $M_{N}(\CM)$. Given the rule~\eqref{Eq:PiRightReg}, the right regular representation of an element $h$ as in Eq.~\eqref{Eq:Elem1} has matrix elements
\begin{equation}\label{Eq:MatElem}
\langle i,  \big ( u^m v^\beta \big) \cdot \big (u^n v^\alpha \big)^{-1} | \pi(h) |j, u^m v^\beta \rangle = w_{ij}(n,\alpha).
\end{equation}
By sampling all allowed values of the labels $i$, $j$, $(n,\alpha)$ and $(m,\beta)$, Eq.~\eqref{Eq:MatElem} returns all the matrix elements of the element $h$.

\section{Topological dynamics by \texorpdfstring{$K$-theory}{K-theory}} 
\label{Sec:TopologicalDynamics}

In $K_0$-theory, two projections from $\KM \otimes C^\ast_r(F)$ are declared equivalent if one can find a continuous family of projections inside the same algebra, interpolating between the two projections. Note that, during this interpolation, the number of engaged local degrees of freedom can change and, as such, the mentioned interpolations are rightfully called stable homotopies. Almost tautologically, the equivalence class $[P]_0$ of a projection $P \in \KM \otimes C^\ast_r(F)$ defines the complete topological invariant associated to the projection relative to stable homotopies. Any projection from $\KM \otimes C^\ast_r(F)$ derives from the spectral projection of a dynamical matrix for a material from the class introduced in section~\ref{Sec:SymmetricPatterns} and, reciprocally, any projection from $\KM \otimes C^\ast_r(F)$ can serve as the band projection of a properly designed material from the same class. Thus, the complete topological invariants associated to the band projections of dynamical matrices are account for by the $K_0$-theory of $C^\ast_r(F)$. The latter accepts a natural abelian semigroup structure and its enveloping group often admits a small number of generators. As already advertised in our introductory remarks, by identifying these generators, one gains complete knowledge and control of the dynamical features supported by our class of architected materials that are robust against the allowed deformations of the materials.

The goal of this section is to advertise the Baum-Connes machinery \cite{BaumIHES1982,BaumProceedings1988,Baum1987,BaumContMath1994}, which is one of the most efficient tools available for mapping the $K$-theories of group $C^\ast$-algebras. This will help us to map the isomorphism classes of the $K_0$-groups of the frieze group algebras. Additionally, we will collect results scattered in the literature in order to supply concrete presentations of the generators $K_0$-groups.  

\begin{remark}{\rm In this work we will not deal at all with $K_1$-theories, which become relevant if the so called chiral symmetry is imposed on the dynamical matrices \cite{ProdanSpringer2016}. This, however, will require a completely different set of experimental conditions.
				
On a similar note, it would generally appear natural to classify materials which are symmetric under a group action using (twisted) equivariant K-theory \cite{FreedAHP2013}, where representatives come from Hamiltonians that are symmetric under a (projective) representation of the symmetry group and the equivalence relation is modified to equivariant homotopy. In this work, non-equivariant K-theory is the correct choice since the invariance under the frieze group is merely an automatic consequence of the protocol that generates the architected material. However, the finer classification of equivariant K-theory would become relevant if one could control the experimental setup such as to impose symmetry also under some representation of the frieze group which acts non-trivially on the internal modes of the resonators.
}$\Diamond$
\end{remark}

\subsection{$K$-theories by Baum-Connes machinery}
\label{Sec:BC} 

The Baum-Connes assembly map $\mu$ makes the connection 
\begin{equation}\label{Eq:BC}
K_0^F(\underline{E}F) \stackrel{\mu}{\to} K_0\big(C_r^\ast(F)\big)
\end{equation}
between the equivariant $K$-homology of a universal space $\underline{E}F$ for proper actions of the frieze group $F$ \cite{BaumContMath1994} and the $K_0$-group of the group $C^\ast$-algebra of $F$. The map $\mu$ in Eq.~\eqref{Eq:BC} is a group isomorphism because the frieze groups are amenable and every countable amenable discrete group is a-T-menable, in which case the Baum-Connes conjecture is known to hold \cite{HigsonKasparov1997}. 

Our first task is to decide on a model for $\underline{E}F$. The full group ${\rm Iso}(\Ss)$ of isometries of the strip is the Lie group  given by the semidirect product of the Lie group $\RM$ of translations and the finite group $\ZM_2 \times \ZM_2$ generated by the two reflections against the two axes of the strip. It has a finite number of connected components. Then, according to \cite{BaumContMath1994}[p.~8], the space of the left cosets 
\begin{equation}
{\rm Iso}(\Ss)/\ZM_2 \times \ZM_2 : = \{s \cdot \ZM_2 \times \ZM_2, \ s \in {\rm Iso}(\Ss)\}  \simeq \RM
\end{equation} 
can serve as a universal space for proper actions of ${\rm Iso}(\Ss)$ and also for any of its discrete subgroups. Therefore, $\underline{E}F$ can be chosen as $\RM$, with the action of $F$ induced from its left action on ${\rm Iso}(\Ss)/\ZM_2 \times \ZM_2$. Specifically, $\tau$ and $\sigma_v$ act as expected while $\sigma_h$ acts as the identity map.

\vspace{0.1cm}

Our next task is to compute the abelian groups $K_0^F(\underline{E}F)$ and the main tool here is the equivariant Chern character \cite{Baum1987,Lueck2002}, which supplies the isomorphism described in the following statement:

\begin{theorem}[Th.~6.1, \cite{Mislin03}] Let $X$ be a proper $F$-CW complex. Then
\begin{equation}\label{Eq:Chern}
  K_0^F(X) \otimes \QM \cong \bigoplus_{r\in \NM} \bigoplus_{\gamma \in S F} H_{2r}(X^{\gamma}/Z(\gamma); \QM),
\end{equation}
where $S F$ is a set of representative members: one for each conjugacy class for each finite order element in $F$. 
$X^\gamma$ is the space of fix points of $X$ under the action of $\gamma$ and $Z(\gamma)$ is the centralizer of $\gamma$. The groups appearing on the right are the ordinary homology groups of spaces with constant coefficients in $\QM$.
\end{theorem}

\begin{remark}{\rm As already stressed in our introduction, the tensoring with $\QM$ washes out the torsion part of $K_0^F(X)$. However, for our concrete context, we know from other sources that the $K_0$-groups, hence also $K_0^F(\underline E F)$, are torsion free. In this case, Eq.~\eqref{Eq:Chern} holds without tensoring by $\QM$ and the K-homology groups can be taken with coefficients in $\ZM$.
}$\Diamond$
\end{remark}

Based on the above observation and the fact that $\mu$ is an isomorphism in our specific context, we conclude
\begin{equation}\label{Eq:FinalChern}
  K_0\big (C^\ast_r(F)\big ) \cong \bigoplus_{r\in \NM} \bigoplus_{\gamma \in S F} H_{2r}(\underline{E}F^{\gamma}/Z(\gamma)).
\end{equation}
As we shall see, the resulting fix point spaces $\underline{E}F^\gamma$ ($\underline{E}F=\RM$ as topological space) in our applications of the Baum-Connes machinery are homeomorphic either to $\SM^1$ or to spaces contractible to a point, hence it is useful to recall that
\begin{equation}\label{Eq:HH}
  H_k(\SM^1) =
  \left\lbrace 
  \begin{array}{ll}
  \ZM, & k =0; \\
  \ZM, & k = 1; \\
  0,   & \text{otherwise}.
  \end{array}
  \right. \qquad
 H_k (\bullet ) =
 \left\lbrace 
 \begin{array}{ll}
 \ZM, & k =0; \\
 0,   & \text{otherwise}.
 \end{array}
 \right.
\end{equation}

The computations of the right-hand side of Eq.~\eqref{Eq:Chern} are detailed below and the results are summarized in Table~\ref{TK0}. These computations are complemented with explicit lists of generators of the $K_0$-groups.

\begin{table}
  \caption{\small Summary of the $K$-theoretic calculations.}
  \centering
  \begin{tabular}{cccc}
    \toprule
  Iso. class & Abstract generators & $S F$ & $K_0(C^\ast_r F) $
  \\ \midrule
  $\mathrm{Z}_\infty$ & $u$  & $\id$ & $\mathbb{Z}$ \\ 
  $\mathrm{Dih}_\infty$ & $u$, $v_1$ & $\id, v_1, uv_1$ & $\mathbb{Z}^3$ \\ 
  $\mathrm{Z}_\infty \times \mathrm{Dih}_1$ & $u$, $v_2$ &  $\id, v_2$ & $\mathbb{Z}^2$ \\ 
  $\mathrm{Dih}_\infty \times \mathrm{Dih}_1$ & $u$, $v_1$, $v_2$ & 
  $\id, v_1, v_2, v_1v_2, uv_1, u v_1 v_2$ & $\mathbb{Z}^6$\\ 
  \end{tabular}
  \label{TK0}
\end{table}

\subsubsection{$\mathrm{Z}_\infty$ isomorphism class}

The groups $p1$ and $p11g$ are both isomorphic to the infinite cyclic group $\mathrm{Z}_\infty = \langle u \rangle \cong \mathbb{Z}$. There are no non-trivial finite order elements since $u^m = 1$ implies $m=0$. For both $p1$ and $p11g$, we have $\underline{E}F^{1}/Z(1)= \mathbb{R}/ \mathbb{Z} \cong \SM^1$ and Eq.~\eqref{Eq:FinalChern} together with \eqref{Eq:HH} give
\begin{equation}
  K_0\big(C^\ast_r (p1) \big ) \cong K_0\big(C^\ast_r (p11g) \big ) \cong \ZM.
\end{equation}
The generators of the $K_0$-groups can be represented by the identity element of the group $C^\ast$-algebras.

\subsubsection{$\mathrm{Dih}_\infty$ isomorphism class}

The groups $p1m1$, $p2$ and $p2mg$ are all isomorphic to the infinite dihedral group $\mathrm{Dih}_\infty$.
Their $K_0$-groups will thus be  isomorphic as well.
We recall the presentation
\begin{align}
  \mathrm{Dih}_\infty = \braket{ u, v_1 }{v_1^2 , (u v_1)^2},
\end{align}
and the fact that every group element can be written as $\gamma_{n,m} = u^n v_1^m$, where $n \in \ZM$ and $m \in \lbrace 0,1\rbrace$.
Using the composition rule from Eq.~\ref{Eq:Compo1}, we obtain
\begin{align}
  \gamma_{n,m}^2 = u^n v_1^m u^n v_1^m
  = u^{n + (-1)^m n} v_1^{ 2 m~ {\rm mod}2}
  =
  \left\lbrace
  \begin{array}{ll}
    u^{2n}, &  m=0 \\
    \id, & m=1
  \end{array}
  \right.
\end{align}
Therefore, $\gamma_{n,m}$ is of finite order whenever $m=0$ and $n=0$ or when $m=1$ for arbitrary $n$.
In those cases, we have $\gamma_{0,0} = \id $ and $\gamma_n : =  \gamma_{n,1} = u^n v_1$. Furthermore,
\begin{align}
  \gamma_{p,q}^{-1} \gamma_n \gamma_{p,q}
  & = v_1^{-q} u^{-p}  u^n v_1 u^p v_1^q
  \notag \\
  & = v_1^{q}   u^{n-2p} v_1^{q+1}
  \notag \\
  & =    u^{(-1)^q (n-2p)} v_1 .
\end{align}
In particular, this means that we have the equivalence by conjugacy
\begin{equation}
  \gamma_n \sim \gamma_{n+2m} .
\end{equation}
Therefore, every finite order element is either conjugate to $1$, to $\gamma_0$ or to $\gamma_1$, and so $S(\mathrm{Dih}_\infty)= \lbrace \id, \gamma_0=v_1, \gamma_1 = u v_1 \rbrace$.
For $q=0$, $\gamma_{p,q}$ is in the centralizer of $\gamma_n$ only if $p=0$, while for $q=1$,  $\gamma_{p,q}$ is in the centralizer of $\gamma_n$ only if $p=n$.
Therefore,
\begin{align}
  Z(\gamma_n) = \lbrace \id, \gamma_n \rbrace \cong \ZM_2 .
\end{align}
The fixed point spaces are $\underline{E}F^{\id} = \RM$, $\underline{E}F^{\gamma_0} = \{ 0 \}$, $\underline{E}F^{\gamma_1}= \{ 1/2 \}$, as for the quotient spaces we have $\underline{E}F^{ \id}  / Z(\id) = [0,1/2]\cong \lbrace\bullet\rbrace$ and obviously this is also true for the rest of them, $ \underline{E}F^{\gamma_0}  / Z(\gamma_0) \cong    \underline{E}F^{\gamma_1}  / Z(\gamma_1) \cong \lbrace\bullet\rbrace$. 
Lastly, Eqs.~\eqref{Eq:FinalChern} and \eqref{Eq:HH} give 
\begin{equation}
  K_0\big(C^\ast_r( \mathrm{Dih}_\infty)\big) 
  \cong  \ZM^3.
\end{equation}

The computation of the above $K_0$-group and its generators appeared in \cite{Cuntz82} as an application  of the methods developed for free products, given that $C^\ast_r(\mathrm{Dih}_\infty))=C_r^\ast(\ZM_2) \star C_r^\ast (\ZM_2)$. These generators have been also identified in \cite{HadfieldJOP2004} by more direct methods. They are:\begin{equation}
  p_1 = \id, \ p_2 = \tfrac{1}{2}(\id-v_1), \ p_3 = \tfrac{1}{2}(\id- u v_1).
 \end{equation}

\subsubsection{$\mathrm{Z}_\infty \times \mathrm{Dih}_1$ isomorphism class}

Only the group $p11m$ belongs to this isomorphism class. We recall the presentation
\begin{align}
  \mathrm{Z}_\infty \times \mathrm{Dih}_1 = \braket{ u, v_2 }{v_2^2 ,  u v_2 u^{-1} v_2 }.
\end{align}
Every group element can be written as $\gamma_{n,m} = u^n v_2^m$, where $n \in \ZM$ and $m \in \lbrace 0,1\rbrace$.
Via Eq.~\ref{Eq:Compo1}, we now obtain
\begin{align}
  \gamma_{n,m}^2 = u^n v_2^m u^n v_2^m
  =
 u^{2n} v_2^{2m}
 = u^{2n} .
\end{align}
Therefore, $\gamma_{n,m}$ is of finite order only if $n=0$, and we have
\begin{equation}
  S  (\mathrm{Z}_\infty \times \mathrm{Dih}_1)
  = \lbrace 
  \id, v_2 \rbrace . 
\end{equation}
Furthermore, we find
\begin{align}
  \gamma^{-1}_{p,q}  v_2 \gamma_{p,q} 
  &=  v_2^{-q} u^{-p}  v_2 u^p v_2^q
  =  v_2^{-q}   v_2  v_2^q
  = v_2 ,
\end{align}
which implies that $Z(v_2) = \mathrm{Z}_\infty \times \mathrm{Dih}_1$.
Since $\underline{E}F^{\id} =\underline{E}F^{v_2} = \RM$, we obtain 
\begin{equation}
  \underline{E}F^\id / Z(\id)  =\underline{E}F^{v_2} / Z(v_2) =  \SM^1.
\end{equation}
Again, employing Eqs.~\eqref{Eq:FinalChern} and \eqref{Eq:HH}, we find
\begin{equation}
  K_0\big(C^\ast_r (\mathrm{Z}_\infty \times \mathrm{Dih}_1)\big)  \cong  \ZM^2.
\end{equation}
The generators of this $K_0$-group can be represented by
\begin{equation}
p_1 = \id, \ p_2 = \tfrac{1}{2}(\id-v_2).
\end{equation}

\subsubsection{$\mathrm{Dih}_\infty \times \mathrm{Dih}_1$ isomorphism class}

Only the group $p2mm$ belongs to this isomorphism class, and we have the presentation
\begin{equation}
  \mathrm{Dih}_\infty \times \mathrm{Dih}_1 =  
  \langle u,v_1,v_2 ~|~ v_j^2, (v_1 v_2)^2, (u v_1)^2, u v_2 u^{-1} v_2\rangle
\end{equation}
Every group element can be written in the form $\gamma_{k,l,m} = u^k v_1^l v_2^m$, where $k \in \ZM$ and $l, m\in \lbrace 0, 1\rbrace $.
In this case, we find
\begin{align}
  \gamma_{k,l,m}^2 
  =
  \left\lbrace
  \begin{array}{ll}
    u^{2k}, &  l=0 \\
    \id, & l=1.
  \end{array}
  \right.
\end{align}
Thus, $\gamma_{k,l,m}^2 = \id$ if either (i) $l=1$ for arbitrary $k$ and $m$ or (ii) $l=0$ and $k=0$ for arbitrary $m$. First, we take a look at case (i) and determine the conjugacy classes and centralizers. Similar to before, we have
\begin{align}
  \gamma_{p,q,r}^{-1}
  \gamma_{k,1,m} \gamma_{p,q,r} 
  & = \gamma_{(-1)^q(k-2p),1,m}.
\end{align}
Therefore, we again obtain a conjugacy equivalence of the form
\begin{equation}
  \gamma_{k,1,m} \sim  \gamma_{\pm(k-2p),1,m} 
\end{equation}
We select $k=0$ and $k=1$ as representative elements and obtain the centralizers
\begin{equation}
Z(\gamma_{0,1,m}) = \lbrace
  1, v_1, v_2, v_1 v_2 \rbrace, \quad
  Z(\gamma_{1,1,m}) = \lbrace
  1, u v_1, v_2, u v_1 v_2 \rbrace.
\end{equation}
Both are isomorphic to the Klein 4-group.
The fix-point sets are
\begin{equation}
\underline{E}F^{\gamma_{0,1,0}} = \underline{E}F^{\gamma_{0,1,1}} = \{0\}, \quad \underline{E}F^{\gamma_{1,1,0}} = \underline{E}F^{\gamma_{1,1,1}} = \{1/2\}.
\end{equation}
In each case, the quotient space $\underline{E}F^\gamma / Z(\gamma)$ is contractible to a point.

\vspace{0.1cm}

Next, we turn to case (ii) and find
\begin{equation}
  \gamma_{p,q,r}^{-1}
  \gamma_{0,0,m}
  \gamma_{p,q,r}
  =
  \gamma_{0,0,m} ,
\end{equation}
and so $Z(\gamma_{0,0,m}) = \mathrm{Dih}_\infty \times \mathrm{Dih}_1 $ .
The fix-point set are
\begin{align}
  \underline{E}F^{\gamma_{0,0,0}} = \underline{E}F^{\gamma_{0,0,1}} = \RM,
\end{align}
such that
\begin{equation}
  \underline{E}F^{\gamma_{0,0,0}} / Z(\gamma_{0,0,0} ) = \underline{E}F^{\gamma_{0,0,1}} / Z(\gamma_{0,0,1} ) = [0, 1].
\end{equation}
Both quotient spaces are contractible to a point.

\vspace{0.1cm}

In summary, we have found
\begin{equation}
  S ( \mathrm{Dih}_\infty \times \mathrm{Dih}_1) = \lbrace  1, v_1, v_2, v_1 v_2, u v_1, u v_1 v_2 \rbrace 
\end{equation}
and all quotient spaces $\underline{E}F^{\gamma} / Z(\gamma )$ are contractible to a point for $ \gamma \in S ( \mathrm{Dih}_\infty \times \mathrm{Dih}_1) $. Then Eqs.~\eqref{Eq:FinalChern} and \eqref{Eq:HH} give
\begin{equation}
  K_0\big (C^\ast_r (\mathrm{Dih}_\infty \times \mathrm{Dih}_1)\big ) \cong  \ZM^6.
\end{equation}

The six generators of the $K_0$-group can be chosen as
\begin{equation}\label{Eq:K0Generators}
  \begin{aligned}
  & p_1 = \id, \ p_2 = \tfrac{1}{2}(\id-v_1), \ p_3 = \tfrac{1}{2}(\id- u v_1), \\
  & p_4 = \tfrac{1}{2}(\id+v_2), \ p_5 = p_2(\id- p_4), \ p_6 = p_3 (\id-p_4).
  \end{aligned}
 \end{equation}
This can be seen by applying K\"unneth theorem  \cite{BlackadarBook}[Th.~23.1.2] to $C^\ast_r (\mathrm{Dih}_\infty \times \mathrm{Dih}_1) \simeq C^\ast_r (\mathrm{Dih}_\infty) \otimes C^\ast_r (\langle \id,v_2\rangle)$.

\section{Algorithm for computing the K-theoretic labels}
\label{Sec:NumInv}

Given a band projection $\pi(p)$, its class in $K_0$-theory can be expressed in one of the bases $\{p_a\}$ listed in the previous section, $[p]_0 = \sum_a n_a [p_a]_0$, $n_a \in \ZM$. From general considerations, we can be sure that two band projections fall into the same topological class if their K-theoretic labels $\{n_a\}$ all coincide. Our goal for this section is to supply practical algorithms for computing the K-theoretic labels of a band projection, generated by an arbitrary dynamical matrix. For this, we employ Kasparov's bivariant K-theory and in particular Kasparov's product \cite{KasparovJSM1981,KasparovJSM1987} for guidance and as a convenient book-keeping instrument, as well as to illustrate a model calculation that can be repeated for other group $C^\ast$-algebras. We use standard notation, {\it e.g.} as in \cite{BlackadarBook}, and a brief review of the concepts is supplied below.

\subsection{General considerations}\label{Sec:GeneralStr} For ungraded $C^\ast$-algebras $A$ and $B$, the elements of Kasparov's abelian group $KK_0(A,B)$ can be presented as the homotopy class of a pair of morphisms $\bar \phi =(\phi_1,\phi_2)$ from $A$ to ${\rm End}^\ast\big (\ell^2(\NM,B) \big )$ with the property that $\phi_1 -\phi_2$ takes values in $\KM \otimes B$ \cite{CuntzKTheory1987}. In particular, $KK_0(\CM,A)$ coincides with the K-group $K_0(A)$ and any morphism $\psi : A \to B$ supplies an element of $KK_0(A,B)$. Furthermore, if $A$, $B$, $C$ are $C^\ast$-algebras, there is an associative product
\begin{equation}
KK_0(A,B) \times KK_0(B,C) \to KK_0(A,C)
\end{equation}
with quite exceptional properties \cite{Cuntz2005}. In particular, if $\psi : A \to B$ and $\psi' : B \to C$ are $C^\ast$-algebra morphisms, then
\begin{equation}
[\psi ] \times [\psi'] = [\psi' \circ \psi].
\end{equation}
This particular case will actually cover all applications considered here. 

\vspace{0.1cm}

Now, let $F$ be any of the frieze groups. In the KK-language, our task is to construct {\it all} pairings of the type
\begin{equation}\label{Eq:Pairing1}
KK_0\big (\CM,C^\ast_r(F)\big ) \times KK_0\big (C^\ast_r(F), \CM\big ) \to KK_0(\CM,\CM) \simeq \ZM,
\end{equation}
such that
\begin{equation}
[p_a]_0 \times [\bar \phi_b] = \Lambda \, \delta_{a,b}, \quad \Lambda \in \NM^\times,
\end{equation}
where $\{[p_a]\}$ is one of the K-theoretic bases listed in the previous section. If $[p]_0 = \sum_a n_a [p_a]_0$, $n_a \in \ZM$, is the decomposition of the $K_0$-class of a projection in such basis, then
\begin{equation}
\Lambda^{-1} \, [p]_0 \times [\bar \phi_a] = n_a
\end{equation}
and this will supply an algorithm to map the K-theoretic labels $\{n_a\}$, hence, all the topological invariants associated to a projection $p \in \KM \otimes C_r^\ast(F)$.

Our strategy is to construct the pairings~\eqref{Eq:Pairing1} in two steps,
\begin{equation}\label{Eq:Pairing2}
\Big (KK\big (\CM,C^\ast_r(F)\big ) \times KK\big (C^\ast_r(F), C^\ast_r(\tilde F)\big ) \Big )\times KK\big (C^\ast_r(\tilde F), \CM\big ) \to \ZM,
\end{equation}
where $\tilde F$ is a finite group.  The reasoning behind the stated strategy is that the pairings in the big parentheses land in $KK\big (\CM,C^\ast_r(\tilde F)\big )$ and the pairings  $KK\big (\CM,C^\ast_r(\tilde F)\big )\times KK\big (C^\ast_r(\tilde F), \CM\big )$ are essentially known and computable once the representation theory of the finite group $\tilde F$ is resolved. Hence, the task reduces to finding appropriate Kasparov cycles from $KK\big (C^\ast_r(F), C^\ast_r(\tilde F)\big )$. As we shall see, this cycle will take the form of a simple morphism. 

\subsection{Explicit pairings} As we have seen in section~\ref{Sec:SymmetricPatterns}, all frieze groups can be derived from generators and relations. In particular, they all have $\ZM$, generated by the element $u$ in our notation, as a normal subgroup.

\begin{proposition} Let $F$ be any of the frieze groups and let $k\ZM$ be the subgroup of $F$ generated by $u^k$ for some fixed $k \geq 2$.  Then $k \ZM$ is a normal subgroup of $F$.
\end{proposition}

\begin{proof} We can treat all frieze groups at once by using the conventions stated in subsection~\ref{Sec:frieze}. Indeed, given the multiplication rule~\eqref{Eq:Compo1}, we have 
\begin{equation}
(u^n v^\alpha)^{-1} u^k (u^n v^\alpha) = v^\alpha u^k v^\alpha = u^{(-1)^{\alpha_1} k} \in k \ZM,
\end{equation}
which shows that, indeed, $k \ZM$ is indeed a normal subgroup.
\end{proof}

\begin{corollary} There exist the group morphisms $F \to F/k \ZM$ from the infinite frieze groups to finite groups. The latter correspond to adding the relation $u^k =1$ to the defining relations of the frieze groups. 
\end{corollary}

\begin{remark}{\rm The above statements will also be relevant for our numerical applications. Indeed, the quotient $F/k\ZM$ supplies a finite group and a canonical morphism $F \to F/k \ZM$, which deliver convenient finite size approximations for the infinite lattice models (see subsection~\ref{Sec:NA}).
}$\Diamond$
\end{remark}

According to Proposition~\ref{Pro:LiftMor}, these morphisms lift to morphisms $\varphi_k : C^\ast_r(F) \to C^\ast_r(F/k\ZM)$ of $C^\ast$-algebras, hence they supply Kasparov cycles $$[\varphi_k] \in KK\big ( C^\ast_r(F), C^\ast_r(F/k\ZM) \big ).$$
For elements from $\ell^1(F)$, the morphism $\varphi_k$ can be written down explicitly as
\begin{equation}\label{Eq:PhiE}
\sum_{n\in \ZM} \sum_{\alpha} c(n,\alpha) u^n v^\alpha  \mapsto \sum_{m=0}^{k-1} \sum_{\alpha} \Big (\sum_{r \in k \ZM} c(m+r,\alpha) \Big ) \tilde u^m \tilde v^\alpha,
\end{equation}
where $\tilde u$ and $\tilde v_j$ are the classes of $u$ and $v_j$'s in $F/k\ZM$. Below is an alternative characterization of the morphism $\varphi_k$:

\begin{proposition}\label{Pro:PhiKPiK} Let $\Tt_k$ be the canonical trace on $\CM F/k \ZM$. Then $\Tt_k \circ \varphi_k$ supplies a positive trace on $\CM F$. Let $(\Hh_k,\pi_k)$ be the right GNS representation of $C^\ast_r(F)$ for this tracial state. Then the image $\pi_k(\CM F)$ in $\BM(\Hh_k)$ coincides with the right regular representation of $C^\ast_r(F/k \ZM)$. Hence, $\varphi_k$ can be identified with the map $\pi_k$.
\end{proposition}

\begin{proof} First, let us point out that the trace $\Tt_k$ generates the right regular representation of $C^\ast_r(F/k\ZM)$ and, since $\varphi_k$ is surjective,   $\Hh_k$ coincide with $\ell^2(F/k \ZM)$, the Hilbert space of the right regular representation. Then a direct computation shows that $\pi_k(f) = \pi \big (\varphi_k(f)\big )$ for any $f \in \CM F$, where $\pi$ is the right regular representation of $F/k\ZM$.
\end{proof}

For us, the relevant finite group is $\tilde F : = F/2\ZM$. Now, given a finite group $G$, any of its finite-dimensional representations supply a morphism $C^\ast_r(G) \to \KM$, hence a cycle from $KK\big (C^\ast_r(G),\CM\big )$. On the other hand, the $K_0$-theory of $C^\ast_r(G)$ can be read off from the ring of representations of the group $G$ \cite{LueckStamm}[Th.~3.2]. Then we have canonical pairings $KK\big (\CM,C^\ast_r(G)\big ) \times KK\big (C^\ast_r(G),\CM\big ) \to \ZM$, which in the present context work as follows:

\begin{proposition}
\label{prop:K0labeling}
Let $F$ be any of the frieze groups and $\tilde F$ be as above. If $\tilde p$ is a projection from $M_N(\CM) \otimes C^\ast_r(\tilde F)$, then the map
\begin{equation}\label{Eq:Rep0}
\tilde F \ni \tilde f \mapsto \left .U_{\tilde f}\right |_{\pi(\tilde p) [\CM^N \otimes \ell^2(\tilde F)]}
\end{equation}
supplies a representation of $\tilde F$ by linear maps on the invariant sub-space $\pi(\tilde p) [\CM^N \otimes \ell^2(\tilde F)]$. Let $\{\rho_a\}_{a = \overline {1,|\tilde F|}}$ be a complete set of irreducible representations of $\tilde F$. Then $\{\rho_a\}$ extend to morphisms from $\KM \otimes C^\ast_r(\tilde F)$ to $\KM$ that generate a basis for Kasparov's group $KK\big (C^\ast_r(\tilde F), \CM\big )$. Furthermore, the pairings
\begin{equation}\label{Eq:Pairing2}
\tilde n_a : = [\tilde p]_0 \times [\rho_a] = \frac{1}{|\tilde F|}\sum_{\tilde f \in \tilde F}  {\rm Tr}\big (U_{\tilde f} \pi(\tilde p)\big ) \, {\rm Tr}\big (\rho_a(\tilde f)\big ) \in \ZM,
\end{equation}
supply a full set of $K$-theoretic labels for $\tilde p$.
Above, ``$\times$'' refers to the Kasparov product
\begin{equation}
KK\big (\CM,C^\ast_r(\tilde F)\big ) \times KK\big (C^\ast_r(\tilde F), \CM\big ) \to KK(\CM,\CM) \simeq \ZM.
\end{equation}
\end{proposition} 

\begin{proof} Any projection $\tilde p$ can be seen as a $C^\ast$-algebra morphism from $\CM$ to $M_N(\CM) \otimes C^\ast_r(\tilde F)$. Indeed, the map
\begin{equation}
\CM \ni c \mapsto c \, \tilde p \in M_N(\CM) \otimes C^\ast_r(\tilde F)
\end{equation}
preserve multiplication precisely because of the idempotency of $\tilde p$,
\begin{equation}
c_1 c_2 \mapsto (c_1 \tilde p) (c_2 \tilde p) =c_1 c_2 \, \tilde p^2 = c_1 c_2 \, \tilde p.
\end{equation}
We denote this morphism by the same symbol $\tilde p$. We claim that Eq.~\eqref{Eq:Pairing2} is nothing else but the trace of $({\rm id} \otimes \rho_a) \circ \tilde p$. According to our observation at the beginning of the previous subsection, $({\rm id} \otimes \rho_a) \circ \tilde p$ encodes the Kasparov product $[\tilde p]_0 \times [\rho_a]$, which lands in $KK_0(\CM,\CM)$, and the trace applied on $({\rm id} \otimes \rho_a) \circ \tilde p$ is simply the isomorphism mapping $KK_0(\CM,\CM)$ to $\ZM$. Thus, the proof is complete if we can confirm the above claim. Now, suppose first that $\tilde p$ is from $C^\ast_r(\tilde F)$. Then, if $\tilde p = \sum_{\tilde f} \tilde p_{\tilde f} \, \tilde f$, the coefficients of this expansion are given by (see Remark~\ref{Re:TT})
\begin{equation}
\tilde p_{\tilde f} = \Tt (\tilde f^{-1} \, \tilde p) = \frac{1}{|\tilde F|}{\rm Tr} \big (U_{\tilde f} \, \pi(\tilde p) \big ).
\end{equation}
More generally, if $\tilde p$ is from $M_N(\CM) \otimes C^\ast_r(\tilde F)$, then the coefficients $\tilde p_{\tilde f}$ are from $M_N(\CM)$ and $\tilde p_{\tilde f} = ({\rm id} \otimes \Tt) (\tilde f^{-1}\tilde p)$. Consequently, 
\begin{equation}
{\rm Tr}(\tilde p_{\tilde f}) = \frac{1}{|\tilde F|}{\rm Tr} \big (U_{\tilde f} \, \pi(\tilde p) \big ).
\end{equation}
Then
\begin{equation}
{\rm Tr}\Big (({\rm id} \otimes \rho_a)(\tilde p) \Big ) = {\rm Tr}\Big (\sum_{\tilde f \in \tilde F} \tilde p_{\tilde f} \otimes \rho_a(\tilde f) \Big ) = \frac{1}{|\tilde F|}\sum_{\tilde f \in \tilde F} {\rm Tr}(U_{\tilde f} \pi(\tilde p)\big ) {\rm Tr}\big (\rho_a(\tilde f) \big ),
\end{equation}
and this completes the proof.
\end{proof}

\begin{remark}{\rm The representations $\rho_a$ generate a complete set of mutually orthogonal central projections \cite{SerreBook}[pg.~50]
\begin{equation}
	\tilde p_a = \tfrac{d_a}{|\tilde F|} \sum_{\tilde f \in \tilde F} {\rm Tr}\big (\rho_a(\tilde f)\big ) \, \tilde f^{-1}, \quad d_a ={\rm dim} \, \rho_a,
\end{equation} 
for which
\begin{equation}
	C^\ast_r(\tilde F) = \bigoplus_a \tilde{p}_a C^\ast_r(\tilde F) \tilde{p}_a \simeq \bigoplus_a \rho_a(C^\ast_r(\tilde F)) =  \bigoplus_a M_{d_a}(\CM)
\end{equation}
and a basis of $K_0(C^\ast_r(\tilde F))$ is obtained by choosing a rank-$1$-projection from each block. The K-theoretic labels $\tilde n_a$ in Eq.~\eqref{Eq:Pairing1} are the expansion coefficients of a projection with respect to such a basis.} $\Diamond$
\end{remark}

\begin{remark}{\rm The right side of Eq.~\eqref{Eq:Pairing2} returns the power $n_a$ in the decomposition $\oplus_a \; \rho_a^{n_a}$ of the representation~\eqref{Eq:Rep0} in the irreducible representations. The Kasparov groups and product are used here to place these algebraic relations in a topological context. Indeed, when phrased as above, the pairings communicate much more, specifically, that the $K$-theoretic labels are invariant against continuous stable deformations of $\tilde p$. The Kasparov groups and product also supplies a very convenient book keeping instrument. Lastly, they enable us to place this and the following computations in a framework that can supply guidance in many other situations.
}$\Diamond$
\end{remark}

Now, according to Proposition~\ref{Pro:PhiKPiK}, $\varphi_2$ lands exactly in the right regular representation of $\tilde F$. Hence $\varphi_2$ can be seen as a cycle from $KK\big (C^\ast_r(F), C^\ast_r(\tilde F)\big )$. The following statement, together with the explicit checks of the assumptions, carried out in the next subsection, assures us that $\varphi_2$ is the actual sought cycle, mentioned in our discussion of the general strategy:

\begin{theorem}\label{Th:Main} Let $F$ be any of the frieze groups and $\{p_a\}$ be the basis of the $K_0$-group of $C^\ast_r(F)$, as listed in the previous section. If there is a complete set of irreducible representations $\{\rho_b\}$ of $\tilde{F}=F/2\ZM$ such that the map,
\begin{equation}\label{Eq:PreTI}
K_0\big (C^\ast_r(F) \big ) \ni [p_a]_0 \mapsto \big \{[p_a]_0 \times [\varphi_2] \times [\rho_b] \big \} \in \ZM^{\times |\tilde F|},
\end{equation}
is injective, then we can find a set of representations $\{\chi_a\}$ and corresponding classes $\{[\chi_a]\} \subset KK\big (C^\ast_r(\tilde F), \CM\big )$ resolving the K-theoretic labels, in the sense that
\begin{equation}
[p_a]_0 \times [\varphi_2] \times [\chi_b] = \Lambda \delta_{a,b}, \quad \Lambda \in \NM^\times.
\end{equation}
Consequently, the full set of K-theoretic labels of a projection $p$ from $\KM \otimes C^\ast_r(F)$ can be computed as
\begin{equation}\label{Eq:TI}
n_a = \frac{1}{\Lambda |\tilde F|}\sum_{\tilde f \in \tilde F}  {\rm Tr}\big(U_{\tilde f} \varphi_2(p)\big) \, {\rm Tr}\big(\chi_a(\tilde f)\big ) \in \ZM.
\end{equation}
\end{theorem}

\begin{proof} Consider the matrix ${\rm R}$ with integer entries ${\rm R}_{ba}= \big([p_a]_0 \times [\varphi_2]\big) \times [\rho_b]$, describing a homomorphism $K_0(C_r^*(F))\simeq \ZM^{\Nn} \to \ZM^{\lvert \tilde{F}\rvert}$, where $\Nn$ is the dimension of the corresponding $K_0$-group. By assumption on the injectivity, ${\rm R}$ has a left inverse $M$ with rational entries and, as such, we can find $\Lambda \in \NM^\times$ and a matrix $\tilde M$ with integer entries such that $\sum_c \tilde M_{bc} {\rm R}_{ca} = \Lambda \, \delta_{ab}$. Setting
\begin{equation}\label{Eq:ChiTilde}
\tilde \chi_b(f) = \Lambda^{-1} \sum_{c} {\rm Tr}(\rho_c(f)) \tilde M_{bc} = \sum_{c} {\rm Tr}(\rho_c(f)) M_{bc},
\end{equation}
then, for any projection $p \sim p_a^{\otimes n_a}$, one has
\begin{align*}
\frac{1}{|\tilde F|}\sum_{\tilde f \in \tilde F}  {\rm Tr}\big(U_{\tilde f} \varphi_2(p)\big) \, \tilde \chi_b(\tilde f) &= \frac{1}{|\tilde F|} \sum_{\tilde f \in \tilde F} n_a  {\rm Tr}\big(U_{\tilde f} \varphi_2(p_a)\big) \, \tilde \chi_b(\tilde f)\\
&= \frac{1}{|\tilde F|} \sum_{\tilde f \in \tilde F} n_a  {\rm Tr}\big(U_{\tilde f} \varphi_2(p_a)\big) \, \tilde \chi_b(\tilde f)\\
&= \frac{1}{|\tilde F|} \sum_{\tilde f \in \tilde F} n_a  {\rm Tr}\big(U_{\tilde f} \varphi_2(p_a)\big) \, \sum_{c} {\rm Tr}(\rho_c(f)) M_{bc}
\\
&=\sum_{c} n_a  {\rm R}_{ca}  M_{bc} = n_a \delta_{ab}.
\end{align*}
Lastly, if we set $\chi_b = \oplus_c \rho_c^{\tilde M_{bc}}$ as the desired representations, then 
\begin{equation}
\tilde \chi_b(\tilde f) = \Lambda^{-1} \, {\rm Tr}\big(\chi_b(\tilde f)\big )
\end{equation}
and the statement follows. \end{proof}

\section{Numerical Experiments}
\label{Sec:Numerical}

\subsection{The generating dynamical matrices}\label{Sec:GenMod} One refers to a set of dynamical matrices as generating models if their band projections supply a complete basis of the $K_0$-group. Of course, one can can always start from $h_j = -p_j$ with $\{p_j\}$ a basis of the $K_0$-group, but sometime this fails because $p_j$'s may involve an infinite number of coefficients in their expansion, {\it i.e.} they don't have  finite range. The latter is essential for laboratory implementations of the models.  While this will not be the case here, let us recall that, in such cases, one simply truncates the series, making sure that the truncation does not close the spectral gap. This opening paragraph serves as a brief reminder of the well known fact that a successful computation of the $K_0$-group and of its generators is equivalent to finding all relevant topological models supported by the algebra of dynamical matrices. Any other dynamical matrix can be deformed into a  stack of elementary models from this generating set.

\vspace{0.1cm}

The present context is special because all model Hamiltonians can be generated directly from the basis of the $K_0$-groups. For example, for the group $p2mm$, the generating dynamical matrices are $h_j = -p_j$, where $p_j$'s are the projections listed in Eq.~\eqref{Eq:K0Generators}. Explicitly,
\begin{equation}\label{Eq:GenModels1}
  h_1 = -\id, \ h_2 = \tfrac{1}{2}(v_1-\id), \ h_3 = \tfrac{1}{2}(u v_1-\id), \ h_4 = -\tfrac{1}{2}(\id+v_2)
  \end{equation}
 and
  \begin{equation}\label{Eq:GenModels2}
 h_5 = \tfrac{1}{4}(v_1 v_2-v_1 -v_2 +\id), \ h_6= \tfrac{1}{4}(u v_1 v_2-uv_1 -v_2 +\id)
 \end{equation}
 
 \begin{figure}
\center
  \includegraphics[width=0.99\linewidth]{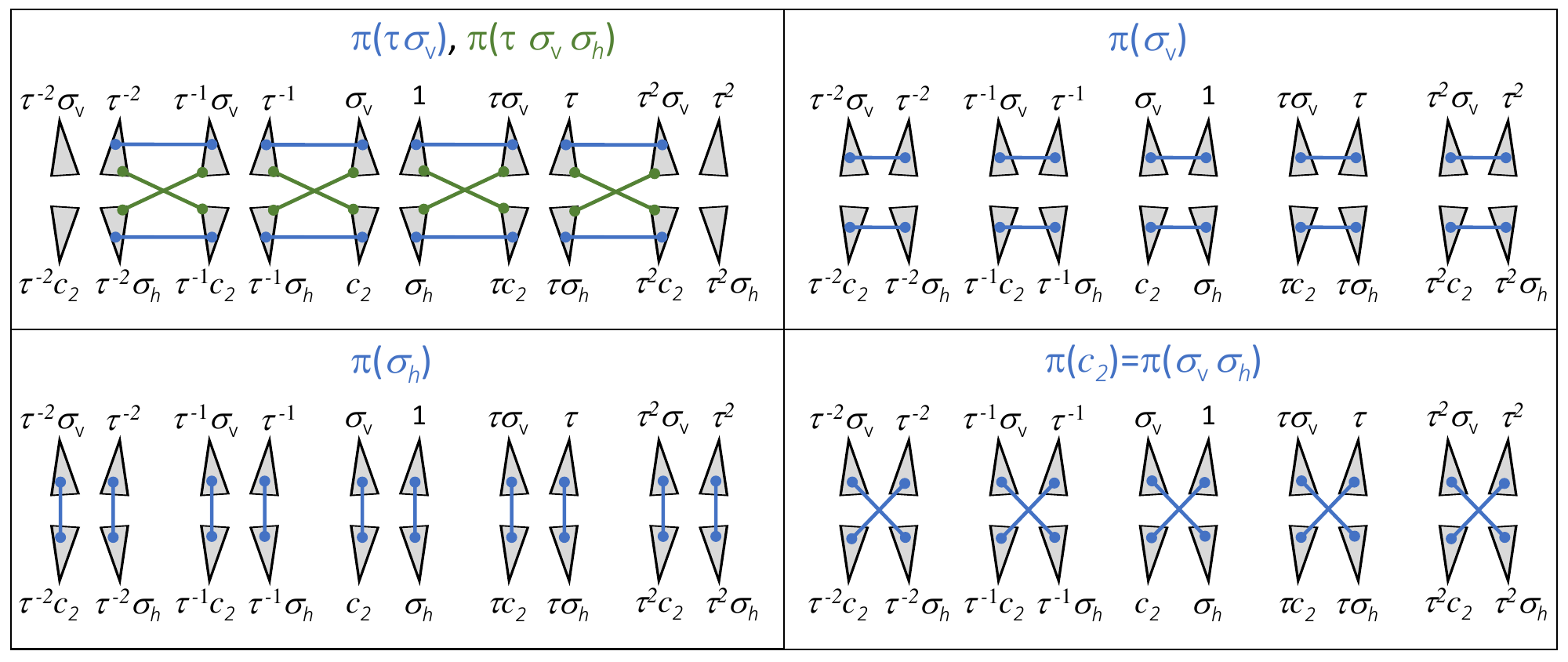}
  \caption{\small The right regular representations of key elements from $C^\ast_r(F)$, $F = p2mm$. The resonators of the patterns are labeled by the group elements of $p2mm$ and a line between two resonators labeled by $f$ and $f'$ represents the hopping operator $|f\rangle \langle f'|$. The right regular representations are the sums of hopping operators displayed in the corresponding panel. In the top-left panel, the elements and the corresponding connections are color coded.}
 \label{Fig:TermActions}
\end{figure}

\vspace{0.1cm}

Their right regular representations can be computed by repeatedly applying the following rule
\begin{equation}
\pi(u^k v^\beta) | u^n v^\alpha \rangle = |u^{n-(-1)^{\alpha_1+\beta_1}k}v^{\alpha+\beta} \rangle.
\end{equation}
These same rules can be easily adopted for the other isomorphism classes of frieze groups, by simply restricting the range of the $\alpha$ coefficients. In fact, focusing on the group $p2mm$ is enough because all the other groups can be embedded in it. In particular, we can list the generating models for the other isomorphism classes: $\{h_1\}$ for $Z_\infty$; $\{h_1,h_2,h_3\}$ for ${\rm Dih}_\infty$; $\{h_1,h_4\}$ for $Z_\infty \times {\rm Dih}_1$.

\vspace{0.1cm}

For the reader's convenience, we report in Fig.~\ref{Fig:TermActions} a graphical illustration of the right regular representations of some key algebra elements that cover all first nearest-neighbor couplings. These diagrams can guide the experimental scientists on how to create model dynamical matrices that belong to a specific topological class. For example, after examining the expressions of the projections~\eqref{Eq:K0Generators}, one should conclude that making the blue couplings in the top-left panel dominant will produce a dynamical matrix from the class of $h_3$. Similarly, for $h_2$, and $h_4$. For a dynamical matrix from the topological class of $h_5$, one should include the couplings shown in the bottom panels and in the top-right panel, hence three types of connections of approximately equal strength. For guidance, we exemplify these ideas in Fig.~\ref{Fig:ModelHamP}. Realistic resonating structures will be supplied in section~\ref{Sec:PlateResonators}. The Mathematica code used to generate Fig.~\ref{Fig:ModelHamP} is available from https://www.researchgate.net/profile/Emil-Prodan.

\subsection{Numerical algorithm for the K-theoretic labels}\label{Sec:NA} We generate finite-size models that can be handled on a computer from the group algebra $C^\ast_R(F/k\ZM)$, for some large even integer $k$. As we already stated, the subgroup $k\ZM$ is normal in $F$, hence the quotient of a self adjoint element $h$ from $F$ to $F/k\ZM$ preserves the resolvent spectrum
\begin{equation}
\RM \setminus {\rm Spec}(h) \subseteq \RM \setminus {\rm Spec}(h \cdot k\ZM).
\end{equation}
In other words, the spectral gaps of $h$ are not contaminated by our finite volume approximations.

 \begin{figure}
\center
  \includegraphics[width=0.99\linewidth]{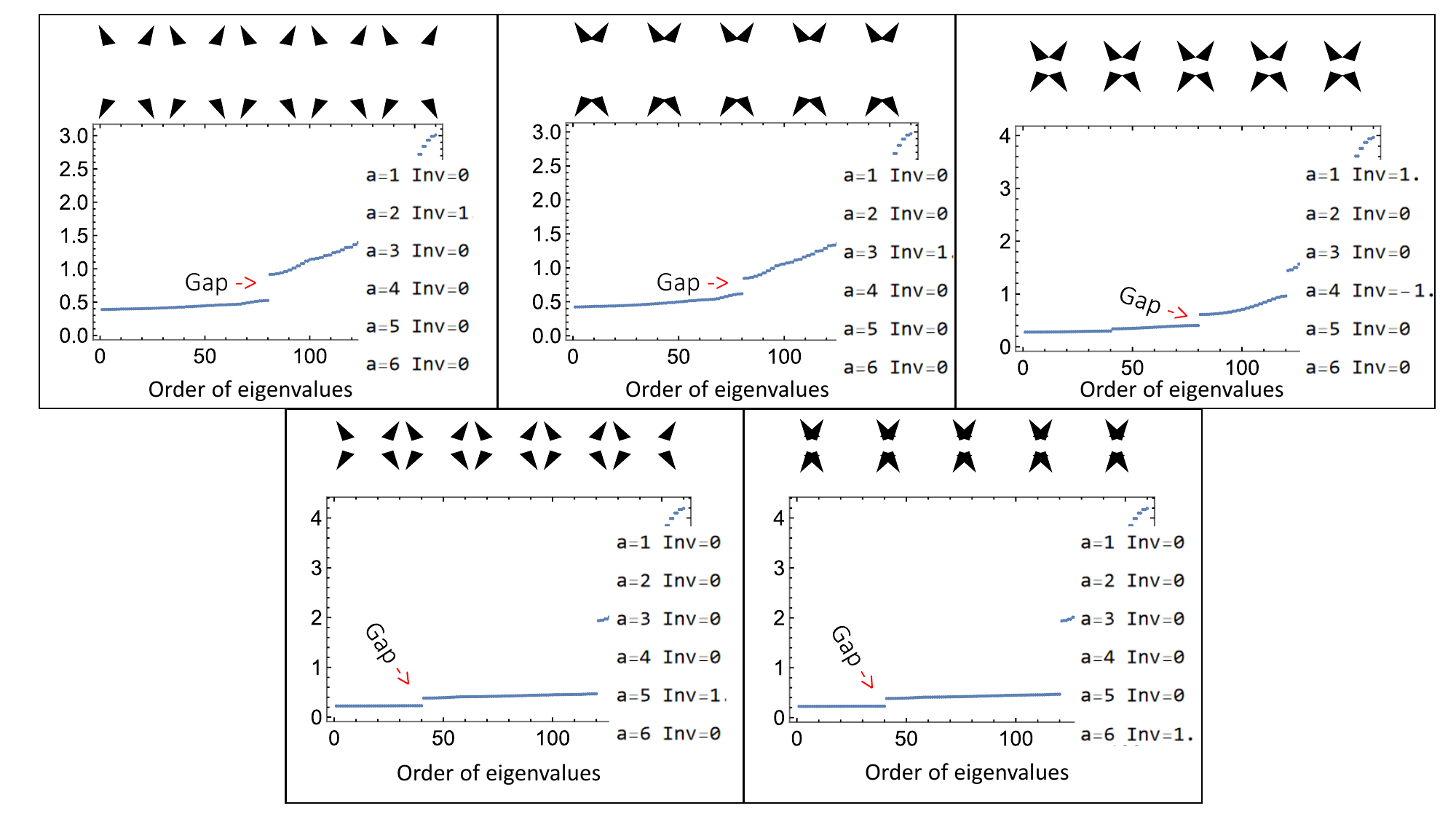}
  \caption{\small Generating dynamical matrices produced from actual patterns, with the resonators assumed to interact as $e^{-2|f\cdot 0-f'\cdot 0|}|f\rangle \langle f'|$, {\it i.e.} with a coupling strength that decays exponentially with the distance between their centers. The values of the six topological invariants associated with the bottom spectral band are displayed in each panel.}
 \label{Fig:ModelHamP}
\end{figure}

First, we need to acknowledge a few practical book keeping details. The group $F/k \ZM$ is manipulated on the computer by mapping it in a convenient model group. This model group is simply a subgroup of the permutation group of $|F/k \ZM|$ objects. For $F=p2mm$, for example, this map is implemented by labeling the group elements as
\begin{equation}
 u^n v_1^{\alpha_1} v_2^{\alpha_2} \mapsto f_j, \quad j = 4n+2\alpha_2 + \alpha_1+1, \quad n = \overline{0,k-1}, \quad \alpha_j =0,1,
\end{equation}
which gives rise to a map $F/k \ZM \to \{1,2,\ldots, |F/k \ZM|\}$. Then the group table can be encoded in a two-index array 
\begin{equation}\label{Eq:Prod}
(j,j') \mapsto {\rm Prod}_{j,j'} \in \{1,\ldots, |F/k \ZM|\},
\end{equation} 
by using the rule~\eqref{Eq:Compo1}. Obviously, the relation stated in Eq.~\eqref{Eq:Prod} assigns a unique permutation to each element $f_j \in F/k\ZM$, acting as
\begin{equation}
f_j(j') ={\rm Prod}_{j,j'}, \quad j' \in \{1,\ldots,|F/k\ZM|\}.
\end{equation}

The basis $|i,f\rangle$ of the Hilbert space corresponding to the right regular representation of $F/k \ZM$ is labeled using the scheme
\begin{equation}
|i,f_j\rangle \mapsto | w \rangle, \quad w=i+ N_s j, \quad i=\overline{1,N_s}, \quad j=\overline{1,|F/k\ZM|},
\end{equation}
and the matrix elements of the representations $\pi(f)$ are stored using the rule~\eqref{Eq:MatElem}, which is applied using the labels and the array ${\rm Prod}$ introduced above. In both rules \eqref{Eq:Compo1} and \eqref{Eq:MatElem}, the sums involving the index $n$ are computed modulo $k$. By setting $k=2$, the above schemes can and are applied to the group $\tilde F$ too.

\vspace{0.1cm}

Now, the first challenge for us is how to implement the map $\varphi_2$. The solution rests on using the spectral projections of the reduced element $u \in F/k\ZM$:
\begin{equation}
P_j = \tfrac{1}{k} \sum_{n =0}^{k-1} e^{-\imath 2 \pi j n/k} \pi(u)^n, \quad j=\overline{0,k-1}.
\end{equation}
Then we define $\varphi_2$ by projecting
\begin{equation}\label{Eq:Phi2}
F/k \ZM \ni f \mapsto \varphi_2(f) = \Pi \pi(f) \Pi, \quad \PI : = P_0 \oplus P_{k/2}.
\end{equation} 
Eq.~\eqref{Eq:Phi2} is justified by the following relations:
\begin{equation}
\pi(u^m) \Pi = \Pi \pi(u^m)= \pi(u)^{m \, {\rm mod} 2} \, \Pi, \quad \pi(v_j) \Pi = \Pi \pi(v_j), 
\end{equation}
which enable us to evaluate
\begin{equation}
\Pi \sum_\alpha \sum_{n=0}^{k-1} c(n,\alpha) \pi\big (u^n v^\alpha \big)  = \Pi \sum_{m=0,1}\sum_{\alpha}\Big (\sum_{r \in 2\ZM} c(m+r,\alpha)\Big ) \pi(u)^m \pi(v)^\alpha 
\end{equation}
and confirm that Eq.~\ref{Eq:Phi2} is equivalent with Eq.~\ref{Eq:PhiE}, once we make the identifications 
\begin{equation}
\tilde u : = \Pi \pi(u) \Pi, \quad \tilde v_j : = \Pi \pi(v_j) \Pi.
\end{equation}

\vspace{0.1cm}

\begin{figure}
\center
  \includegraphics[width=0.99\linewidth]{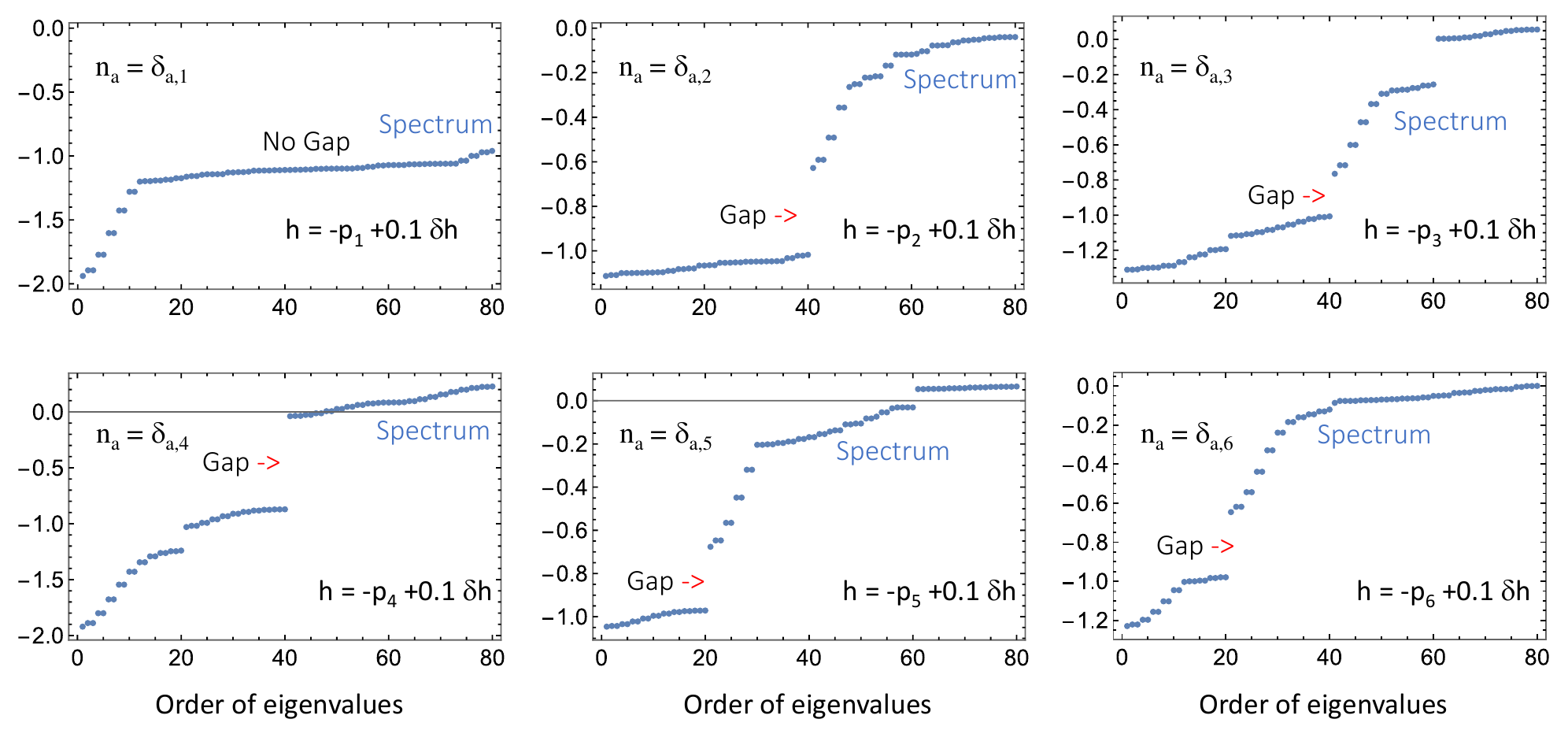}
  \caption{\small Spectrum of perturbations of the generating dynamical matrices listed in Eqs.~\eqref{Eq:GenModels1} and \eqref{Eq:GenModels2}, together with the topological invariants $n_a$, $a=\overline{1,6}$, as computed with Eq.~\eqref{Eq:NTI1} for the gap projections.}
 \label{Fig:GenModels}
\end{figure}

The second and last challenge is how to evaluate the trace ${\rm Tr}\big(U_{\tilde f} \varphi_2(p)\big)$. Switching the order of the operators inside the trace, we have
\begin{equation}
{\rm Tr}\big(\varphi_2(p) U_{\tilde f} \big) = \sum_{i=1}^{N_s}\sum_{\tilde f' \in \tilde F} \langle i, \tilde f' |  \varphi_2(p) U_{\tilde f} |i, \tilde f'\rangle = \sum_{i=1}^{N_s} \sum_{\tilde f' \in \tilde F}\langle i, \tilde f' | \varphi_2(p) |i,\tilde f \tilde f' \rangle.
\end{equation}
Next, we connect with what was said above and represent the elements $\tilde f'$ of $\tilde F$ as $\Pi f$ with $f$ from $F/k\ZM$, in which case,
\begin{equation}
{\rm Tr}\big(U_{\tilde f} \varphi_2(p)\big) = \sum_{i=1}^{N_s} \sum_{f \in F/k\ZM}\langle i, f | \Pi \pi(p) \Pi |i,\tilde f f \rangle,
\end{equation}
where the sum is taken over all elements of $F/k\ZM$. Here, $\tilde f \in F/2 \ZM$ is seen as an element of $F/k \ZM$ via the canonical section of the morphism $F/k \ZM \to F/2 \ZM$. To verify that the normalization constants are properly fixed, we have checked that ${\rm Tr}\big(U_{\tilde 1} \varphi_2(1)\big)=8$, the dimension of the right regular representation of $\tilde F$ when $N_s$ is set to one. 

\vspace{0.1cm}

We now have the numerical tools to examine the map defined in Eq.~\eqref{Eq:PreTI}. We use the numerical formula
\begin{equation}\label{Eq:NTI1}
[p_a]_0 \mapsto \frac{1}{|\tilde F|}\sum_{\tilde f \in \tilde F} \rho_b(\tilde f) \sum_{i=1}^{N_s} \sum_{f \in F/k\ZM}\langle i, f | \Pi \pi(p_a) \Pi |i, \tilde f f \rangle
\end{equation}
and the eight independent one dimensional irreducible representations of $\tilde F$ enumerated below: 
\begin{equation}
\rho : =\Big [\rho_a (\tilde f_b) \Big ]_{a,b=\overline{1,8}} = {\tiny \begin{pmatrix} 1 & 1 & 1 & 1 & 1 & 1 & 1 & 1 \\

 1 & -1 & 1 & -1 & 1 & -1 & 1 & -1 \\

 1 & 1 & 1 & 1 & -1 & -1 & -1 & -1 \\

 1 & -1 & 1 & -1 & -1 & 1 & -1 & 1 \\

 1 & 1 & -1 & -1 & 1 & 1 & -1 & -1 \\

 1 & -1 & -1 & 1 & 1 & -1 & -1 & 1 \\

 1 & 1 & -1 & -1 & -1 & -1 & 1 & 1 \\

 1 & -1 & -1 & 1 & -1 & 1 & 1 & -1  \end{pmatrix} }.
\end{equation} 
They come from the presentation of $\tilde F$ as the direct product $\langle \tilde v_1 \rangle \times \langle \tilde v_2 \rangle \times \langle \tilde u \rangle$ and from the tensor products of the standard representations of the constituent $\ZM_2$ subgroups. In particular, all $\rho_a$ representations are one dimensional. As a check, we can confirm that $\rho \cdot  \rho^T = I_{8 \times 8}$. 

\vspace{0.1cm}

The result for ${\rm R}$ is
\begin{equation}
\big [{\rm R}_{ba}\big] = \Big [[p_a]_0 \times [\varphi_2] \times [\rho_b] \Big ]_{^{b=\overline{1,8}}_{a=\overline{1,6}}}= {\tiny \begin{pmatrix} 
1 & 0 & 0 & 1 & 0 & 0 \\
 1 & 1 & 1 & 1 & 0 & 0 \\
 1 & 0 & 1 & 1 & 0 & 0 \\
 1 & 1 & 0 & 1 & 0 & 0 \\
 1 & 0 & 0 & 0 & 0 & 0 \\
 1 & 1 & 1 & 0 & 1 & 1 \\
 1 & 0 & 1 & 0 & 0 & 1 \\
 1 & 1 & 0 & 0 & 1 & 0
 \end{pmatrix}}
\end{equation}
and the rank of this matrix is six, hence equal to the dimension of the $K_0$-group. This confirms that the assumption in Theorem~\ref{Th:Main} holds. It also reassures us that the six generators $\{p_a\}$ listed in Eq.~\eqref{Eq:K0Generators} are indeed independent. Furthermore, we can produce the left inverse $M=\big ( {\rm R}^T \cdot {\rm R} \big )^{-1} \cdot {\rm R}^T$ of ${\rm R}$ and define the linear maps $\{\tilde \chi_a\}$ introduced in Eq.~\eqref{Eq:ChiTilde},
\begin{equation}
\Big [\tilde \chi_a(\tilde f_b)\Big ]_{^{a=\overline{1,6}}_{b=\overline{1,8}}} = M \cdot \rho = {\tiny \begin{pmatrix}
1 & 1 & -1 & -1 & 0 & 1 & 0 & -1 \\

 0 & -2 & 0 & -2 & 0 & 0 & 0 & 0 \\

 0 & 0 & 0 & 0 & 0 & -2 & 0 & -2 \\

 0 & 0 & 2 & 2 & 0 & 0 & 0 & 2 \\

 0 & 0 & 0 & 4 & 0 & 0 & 0 & 0 \\

 0 & 0 & 0 & 0 & 0 & 0 & 0 & 4
 \end{pmatrix} },
 \end{equation}
which deliver the desired pairings. Lastly, the K-theoretic labels of a projection $p$ are computed numerically via
\begin{equation}\label{Eq:NTI1}
n_a = \frac{1}{|\tilde F|}\sum_{\tilde f \in \tilde F} \tilde \chi_b(\tilde f) \sum_{i=1}^N \sum_{f \in F/k\ZM}\langle i, f | \Pi \pi(p) \Pi |i, \tilde f f \rangle .
\end{equation}

\begin{figure}
\center
  \includegraphics[width=0.99\linewidth]{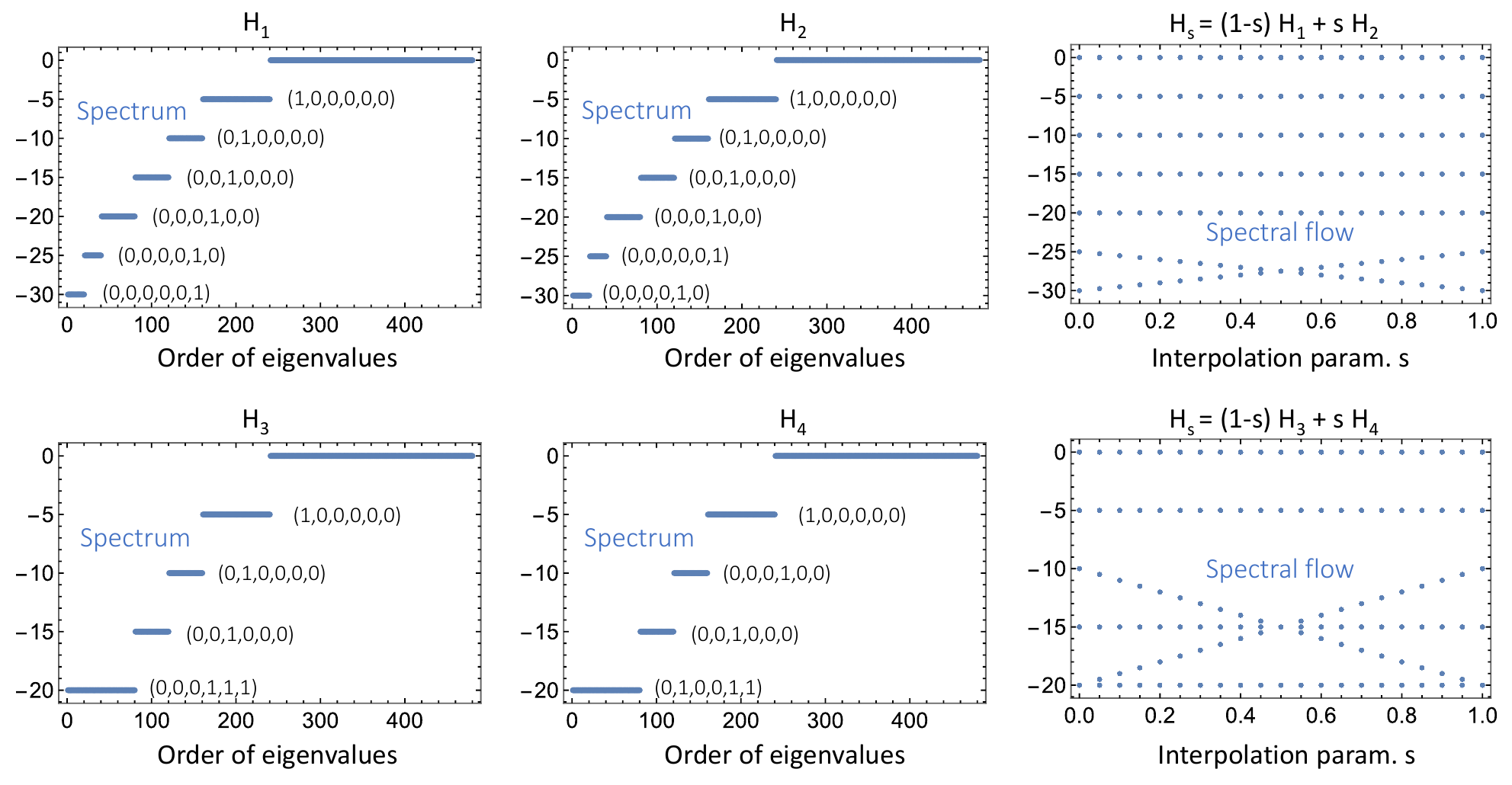}
  \caption{\small Demonstrations of topological spectral flow generated by linear interpolation between models displaying  bands with distinct K-theoretic labels. The spectra of the interpolated dynamical matrices are shown on the right, together with the six topological invariants of the seen spectral bands.}
 \label{Fig:Inter}
\end{figure}

\vspace{0.1cm}

We want to point out that this numerical formula can be canonically adapted to space groups in arbitrary dimensions, once the finite subgroup is identified.

\subsection{Computations of K-theoretic labels} We first validate our numerical algorithm by computing the K-theoretic labels for the band projections of the model dynamical matrices listed in Eqs.~\eqref{Eq:GenModels1} and \eqref{Eq:GenModels2}. To make it somewhat more challenging, we actually added a perturbation $\delta h$, which was generated by populating the series~\eqref{Eq:Elem1} with random entries up to  $|n| \leq 2$ (hence 24 random terms). Fig.~\ref{Fig:GenModels} reports the spectra of these dynamical matrices, together with evaluations of the parings~\eqref{Eq:NTI1} for the gap projections corresponding to the indicated spectral gaps. As one can see, the numerical results confirm the statement made in Eq.~\eqref{Eq:TI}. The Mathematica code used to generate Fig.~\ref{Fig:GenModels} is available from https://www.researchgate.net/profile/Emil-Prodan.

\vspace{0.1cm}

Next, we want to demonstrate observable physical effects predicted with the tools developed by our work. For this, let $e_{i}$ be the matrix with entry 1 at position $(i,i)$. We then generate Hamiltonians from $M_6(\CM) \otimes C^\ast_r(F)$ using the scheme
\begin{equation}\label{Eq:MultiLayer}
h = \sum_{i=1}^6 \Lambda_i \; e_i \otimes h_{\xi(i)}
\end{equation}
where the coefficients $\Lambda_i$ are chosen such that $h$ displays a desired number of bands and $\xi$ is a permutation of the indices, which will enable us to shuffle the bands and to engineer interesting K-theoretic labels.

\vspace{0.1cm}

Fig.~\ref{Fig:Inter} reports spectra of dynamical matrices generated with the scheme outlined above, together with the computed K-theoretic labels. Fig~\ref{Fig:Inter} also illustrates how a linear interpolation between models with different K-theoretic labels prompts the predicted topological spectral flow. The Mathematica code used to generate Fig.~\ref{Fig:Inter} is available from https://www.researchgate.net/profile/Emil-Prodan.

\section{Topological spectral flows with plate resonators}
\label{Sec:PlateResonators}

\begin{figure}
\center
  \includegraphics[width=0.99\linewidth]{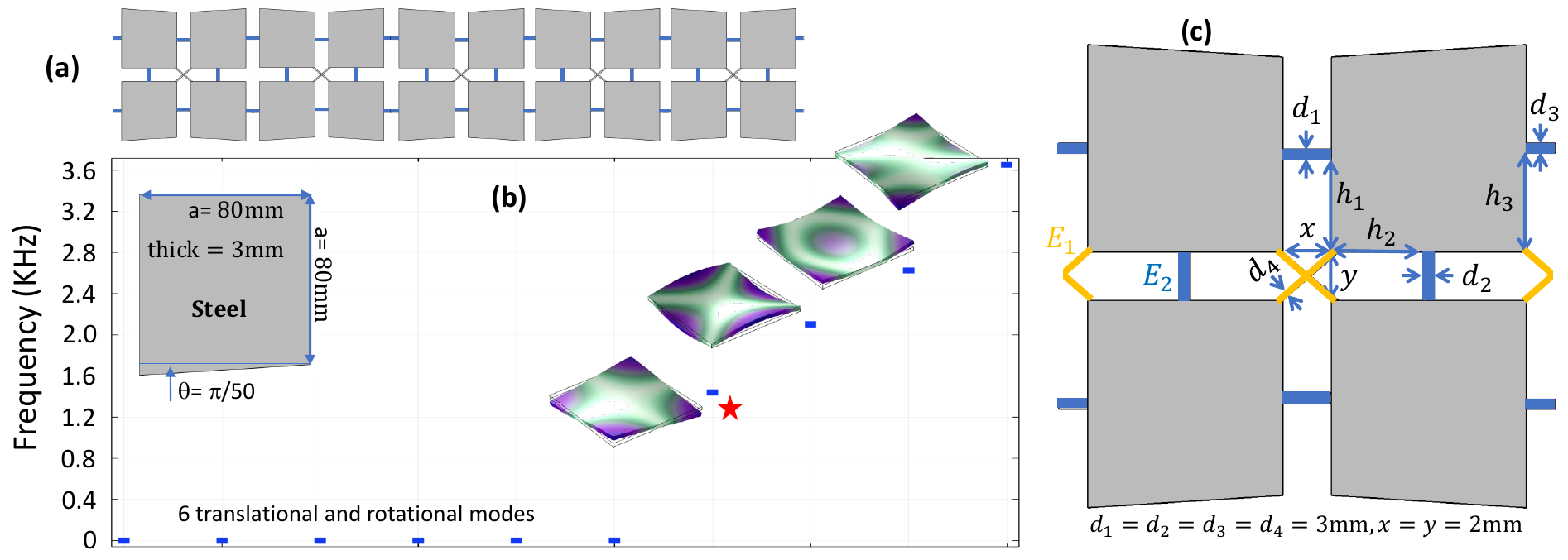}
  \caption{\small (a) A resonating structure of coupled steel plates, generated with the frieze group $p2mm$. (b) The seed resonator, together with its physical characteristics and a representation of its first 10 resonant modes. (c) Details of the couplings, notations and values of the fixed parameters.}
 \label{Fig:BD1}
\end{figure}

In this section, we illustrate how the discrete models discussed so far can be implemented with physical resonating structures. Fig.~\ref{Fig:BD1}(a) describes a pattern of coupled steel plates, generated with the frieze group $p2mm$ from one seed resonator. The basic plate resonator and its physical characteristics are presented in Fig.~\ref{Fig:BD1}(b), together with its first ten resonant modes. We focus on the dynamical features generated from the first flexural mode of the plate, indicated in Fig.~\ref{Fig:BD1}(b) by the star. Note that the geometry of the plate was chosen such that this mode is well separated in frequency from all the other resonant modes. As such, the collective dynamics of the plates in the regime of weak couplings and in a frequency domain around the frequency of this particular mode can be accurately quantified using a discrete model with one resonant mode per resonator.

\vspace{0.1cm}

 Couplings between the resonators are implemented by thin metal bridges that respect the $p2mm$ symmetry. The strength of these couplings can be controlled by the position of these bridges: A bridge placed closer to a node (anti-node) of the resonant mode produces a weaker (stronger) coupling (see the mode's spatial profile in Fig.~\ref{Fig:BD1}(b)). The strength of the couplings can be further adjusted by modifying the Young modulus $E$ of the metal bridges. Thus, our space of adjustable parameters consists of five parameters $(h_1,h_2,h_3,E_1,E_2)$ (see Fig.~\ref{Fig:BD1}(c)).

\begin{figure}
\center
  \includegraphics[width=0.5\linewidth]{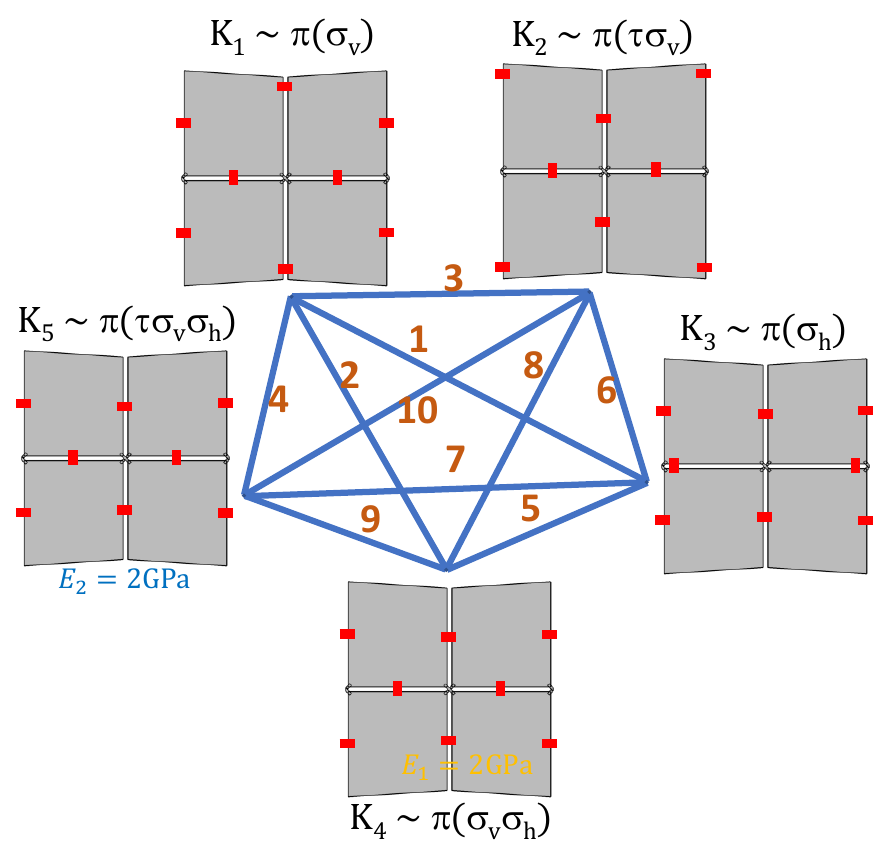}
  \caption{\small Configurations of the couplings generating dynamical matrices from the topological class of the algebra elements mentioned in Fig.~\ref{Fig:TermActions}. The Young modulus of the metal bridges was fixed at $E_1 = E_2=0.1$GPa, except for the two cases indicated in the diagram. The diagram also shows and labels the ten distinct ways in which the physical systems can be interpolated.}
 \label{TenClass}
\end{figure}

\vspace{0.1cm}

In Fig.~\ref{Fig:ModelHamP}, we indicated how the resonators are to be coupled in order to reproduce the fundamental model dynamical matrices. However, in practice, we found that it is much easier to implement dynamical matrices $K_i$, $i = \overline{1,5}$, from the topological classes of the self-adjoint elements presented in Fig.~\ref{Fig:TermActions}. Indeed, for each of those elements, we only need to ensure one specific dominant coupling and this can be achieved with the configurations shown in Fig.~\ref{TenClass}. We have computed the K-theoretic labels of the lowest spectral bands of $K_i$ dynamical matrices, based on the assumption that they indeed belong to topological class of the algebra elements indicated in Fig.~\ref{TenClass}. The results are: $K_1 \rightarrow (0,1,0,0,0,0)$, $K_2 \rightarrow (0,0,1,0,0,0)$, $K_3 \rightarrow (1,0,0,-1,0,0)$, $K_4 \rightarrow (1,1,0,-1,-2,0)$, $K_5 \rightarrow (1,0,1,-1,0,-2)$,
 hence they are all distinct from each other. The Mathematica code used to generate these topological labels is available from https://www.researchgate.net/profile/Emil-Prodan.
 
\vspace{0.1cm}

\begin{remark}{\rm Among other things, the above K-theoretic labels assures us that actually $K_i$ together with $K_0 = 1$ can serve as an alternative generating set of model dynamical matrices. In many respects, this set is more natural than the one we originally considered.
}$\Diamond$
\end{remark}

Each of the resonating structures seen in Fig.~\ref{TenClass} corresponds to a point in the 5-dimensional parameter space $(h_1,h_2,h_3,E_1,E_2)$. Based on the K-theoretic labels computed above, we predict that, no matter how we interpolate between two of these five points, we cannot avoid the closing of the gap in the resonant spectrum. This phenomenon should be extremely robust because the K-theory predicts that the gap closing cannot be avoided even if the structure is connected to external frames, provided this connection does not change the topological class of the initial and final resonating structures. Furthermore, we conjecture that the topological spectral flows are robust against disorder.

\vspace{0.1cm}

The five topologically distinct resonating structures from Fig.~\ref{TenClass} can be pairwise interpolated in ten different ways, as specified there. Each of these interpolations will produce distinct topological spectral flows that can find different practical applications. In Fig.~\ref{Fig:ESF}, we report the COMSOL-simulated resonant spectra for the ten interpolations between different $(i,j)$ pairs from Fig.~\ref{TenClass}, where the parameters were varied linearly as
\begin{equation}\label{Eq:Interp}
\begin{aligned}
\big 
(h_1(s) ,h_2(s), h_3(s), E_1(s), E_2(s)\big)  = & (1-s)(h_1^i ,h_2^i, h_3^i, E_1^i, E_2^i) \\
& +s(h_1^j ,h_2^j, h_3^j, E_1^j, E_2^j).
\end{aligned}
\end{equation}
These simulations fully confirm the K-theoretic predictions and serve as a demonstration of how spectral engineering can be achieved with plate resonators generated by space groups.

\begin{figure}
\center
	\includegraphics[width=0.99\linewidth]{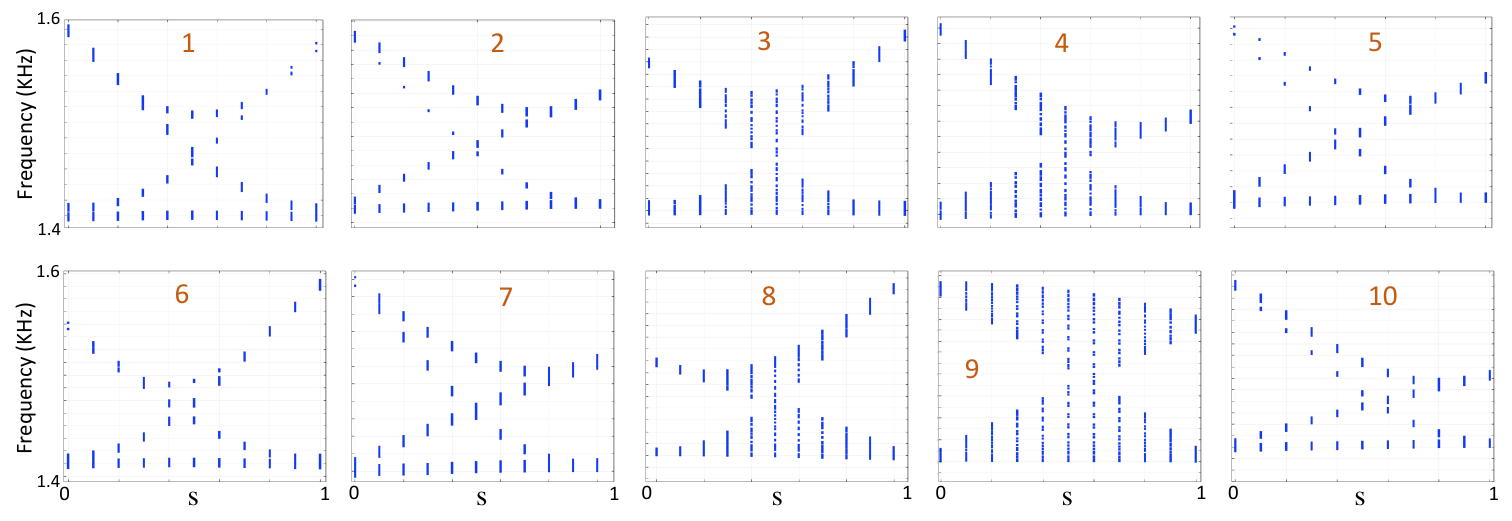}
	\caption{\small (a) Spectral flow of the resonant spectrum when the configuration of the coupled plate resonators is deformed (see Eq.~\eqref{Eq:Interp}) along the ten distinct interpolations shown in Fig.~\ref{TenClass}.}
	\label{Fig:ESF}
\end{figure}


\end{document}